\documentclass[a4paper,11pt]{article}
\usepackage{fullpage}
\usepackage{times}
\usepackage{amsmath,amssymb}
\usepackage{amsthm}
\usepackage{array}
\usepackage{bm}
\usepackage{booktabs}
\usepackage[small]{caption}
\usepackage[hidelinks]{hyperref}
\usepackage{pgfplots}
\usepackage{tikz}
    \usetikzlibrary{decorations.text}
    \usetikzlibrary{shapes,arrows}
    \tikzset{every picture/.style={remember picture}}
\usepackage{upgreek}
\usepackage{url}
\usepackage{wrapfig}
\usepackage{collcell}
\usepackage{authblk}

\renewcommand{\subsubsection}[1]{\smallskip\textbf{#1.~}}

\newenvironment{talign}
 {\align}
 {\endalign}
\newenvironment{talign*}
 {\csname align*\endcsname}
 {\endalign}

\newtheorem{proposition}{\bf Proposition}

\newtheorem{theorem}[proposition]{\bf Theorem}

\newcommand{\BR}{\textit{{BR}}}
\newcommand{\xl}{\textit{{L}}}
\newcommand{\xf}{\textit{{F}}}
\newcommand{\otc}{o}
\newcommand{\vecone}[1]{#1}

\usepackage{natbib}

\renewcommand{\cite}[1]{\citep{#1}}

\begin{document}

\title{Imitative Follower Deception in Stackelberg Games%
\thanks{Jiarui Gan was supported by the EPSRC International Doctoral Scholars Grant EP/N509711/1. Zinovi Rabinovich's contribution to this paper was supported by NTU SUG M4081985 grant.}}

\author[1]{Jiarui Gan}
\author[2]{Haifeng Xu}
\author[3]{Qingyu Guo}
\author[4]{Long Tran-Thanh}
\author[3]{\\ Zinovi Rabinovich}
\author[1]{Michael Wooldridge}

\affil[1]{University of Oxford}
\affil[2]{Harvard University}
\affil[3]{Nanyang Technological University}
\affil[4]{University of Southampton}

\renewcommand\Authands{ and }

\date{}

\maketitle

\begin{abstract}
Information uncertainty is one of the major challenges facing applications of game theory. In the context of Stackelberg games, various approaches have been proposed to deal with the leader's incomplete knowledge about the follower's payoffs, typically by gathering information from the leader's interaction with the follower. Unfortunately, these approaches rely crucially on the assumption that the follower will \emph{not} strategically exploit this information asymmetry, i.e., the follower behaves truthfully during the interaction according to their actual payoffs.
As we show in this paper, the follower may have strong incentives to deceitfully imitate the behavior of a different follower type and, in doing this, benefit significantly from inducing the leader into choosing a highly suboptimal strategy.
This raises a fundamental question: how to design a leader strategy in the presence of a deceitful follower?
To answer this question, we put forward a basic model of Stackelberg games with \emph{(imitative) follower deception} and show that the leader is indeed able to reduce the loss due to follower deception with carefully designed {policies}.
We then provide a systematic study of the problem of computing the optimal leader {policy} and draw a relatively complete picture of the complexity landscape;
essentially matching positive and negative complexity results are provided for natural variants of the model. Our intractability results are in sharp contrast to the situation with no deception, where the leader's optimal strategy can be computed in polynomial time, and thus illustrate the intrinsic difficulty of handling follower deception.
Through simulations we also examine the benefit of considering follower deception in randomly generated games.
\end{abstract}

\section{Introduction}
\label{sec:intro}

Recently, there is a growth of interest in Stackelberg games in the AI community.
This trend is driven in part by a number of high-impact real-world applications in security domains~\cite{tambe2011security}.
Beyond that, Stackelberg game models also find many noteworthy applications in other problems, such as principal-agent contract design \cite{contract4,contract3} and exam design for large-scale tests \cite{li2017game}.
In a Stackelberg game, a leader commits to a mixed strategy and a follower best-responds after observing the leader's strategy. The Stackelberg equilibrium yields the optimal strategy the leader can commit to in this framework.
As in many other game-theoretic applications, a key real-world challenge facing applications of Stackelberg games is that the leader may not have full information about the follower's payoffs for computing the equilibrium.
To address this issue, various approaches have been proposed:
When the leader can {estimate} the follower's payoffs to within certain intervals, leader strategies that are robust against small interval uncertainties are computed \cite{letchford2009learning,kiekintveld2013security,nguyen2014regret};
When the leader knows a distribution about the follower's payoffs/types, Bayesian Stackelberg
equilibria are studied to maximize the leader's average utility~\cite{conitzer2006computing,paruchuri2008playing,jain2008bayesian};
When the leader can interact with the follower, algorithms are designed to learn the optimal leader strategy to commit to from the follower's best responses in the interaction~\cite{letchford2009learning,blum2014learning,haghtalab2016three,roth2016watch,peng2019learning}.

Despite their differences,  the above approaches share a common and crucial step --- obtaining payoff-relevant information about the follower. However, knowing the leader's attempt at information gathering, a strategic follower would be incentivized to intentionally distort information learned by the leader, in particular by deceitfully imitating the behavior of a different follower type --- a phenomenon we term \emph{imitative follower deception}.
Unfortunately, all aforementioned approaches ignore the possibility of a deceitful follower and adopt a simplistic assumption that any follower information the algorithm gathers is truthful.

As we will show in this paper, algorithms designed under the above assumption can be easily manipulated by the follower and may produce highly suboptimal leader strategies as a result.
Consider perhaps the most basic learning setting where the follower has an uncertain \emph{type} (i.e., a full set of payoff parameters) $\theta \in \Theta$. Knowing the set $\Theta$ of all possible types but not the actual $\theta$,  the leader wants to learn the optimal strategy against this follower.
If the follower is truthful, the optimal strategy the leader can learn would be the strong Stackelberg equilibrium. However, if the follower behaves completely according to the payoff parameters of another type $\theta'$ (he may indeed have such an incentive, as we will show), the leader can only learn the optimal strategy against the fake type $\theta'$, which may be highly suboptimal. A fundamental question then is: \emph{what is the optimal strategy the leader can learn or design when facing a deceitful follower?}

In this paper, we put forward a basic Stackelberg model in an attempt to formalize this question and aim to understand how imitative follower deception affects the leader's choice of strategies.
We note that, even though there might be other forms of follower deception as well, imitative deception is arguably the simplest and the most natural deception approach, which we expect to happen the most often in practice.
Our model is Bayesian by nature: the type of the follower is drawn from a finite set $\Theta$; the follower knows his true type, but the leader only has a probabilistic belief about it. To play against a deceitful follower, the leader commits to a \emph{policy} --- a ``menu'' that specifies a mixed strategy to play for each (learned) follower type. Knowing the leader's commitment, a follower of true type $\theta \in \Theta$ can imitate another type $\theta' \in \Theta$, behaving consistently according to the payoffs of $\theta'$ to mislead the leader into learning the fake type $\theta'$. The leader then plays the strategy specified by the committed policy for $\theta'$.
The optimal leader policy maximizes the leader's expected utility, taking into account the follower's optimal imitating strategy in the loop; this is precisely the optimal leader strategy that can be learned in the deceptive setting.

\subsection{Our Contribution}

We make the following contribution.
In Section~\ref{sec:example}, we provide a motivating example to illustrate the mechanisms of imitative follower deception, the loss it brings to the leader, and how the leader can reduce the loss with carefully designed strategies.
Then, we formalize the model in Section~\ref{sec:model}.

In Section~\ref{sec:complexity}, we study the hardness of handling follower deception and draw a relatively complete picture of the problem's complexity landscape.
We consider settings both with and without \emph{incentive compatibility} (IC) constraints; the former requires the leader strategy to incentivize the follower to report his true type while the latter does not. On the negative side, we show that it is hard even to approximate the optimal leader policy within any meaningful ratio unless P = NP, with or without IC constraints; the ratios are also shown to be tight.
The results indicate that deception is the fundamental reason of the intractability as, when there is no deception, the problem can be solved in polynomial time simply by computing the optimal commitment for each follower type separately.
Despite the computational barriers, we provide an MILP (mixed integer linear program) formulation as a practical approach to compute the optimal leader policy, from which a fixed-parameter tractability with respect to the number of follower types can be derived.

In Section~\ref{sec:mixed_policy}, we further observe that the leader can improve her utility using a \emph{mixed policy}, that samples from a distribution a strategy to play for each follower type. We study the respective computational problems above with mixed policies and find that, in addition to the utility improvement they bring, the optimal mixed policy can be computed in polynomial time when IC constraints are imposed; though, the problem remains hard to approximate without IC.

In Section~\ref{sec:experiment}, through simulations, we compare the improvement of leader utility by our approaches and other benchmark approaches in randomly generated games, and examine the benefit of considering follower deception. We conclude in Section~\ref{sec:conclusion} with a short discussion on future directions of this work.

\subsection{Additional Applications and Related Work}

Besides furthering our understanding about the best leader strategy that can be learned when facing a deceitful follower, our model can also be applied directly to a number of applications. For example, in principal-agent contract design (a widely studied Stackelberg model in economics \cite{contract4,contract3,contract1}), a principal may present a menu of contracts to an uncertain agent and ask the agent to choose one that matches his true type. Naturally, the principal's menu must take into account the agent's misreport of his type. Similar ideas apply to exam design for large-scale tests, e.g., for MOOCs (massive open online courses), which have also been modeled as Stackelberg games \cite{li2017game}. In these situations, the tester may need to use different exams for different types of exam takers who usually come from various education backgrounds and have different learning objectives. Again, the exam design will also need to take into account the exam taker's possible misreport about their type  to strategically conceal their strength in the hope of a preferred test.

Deception has been extensively studied in Stackelberg games, particularly in its applications in security domains \cite{Brown2005,Powell2007,zhuang_bier_2010,yin2013optimal,xrtd15,guo2017comparing,schlenker2018deceiving}. However, all previous works focus on designing deceptive strategies for the \emph{leader} (usually, a defender).  In these models, it is the leader who has the informational advantage, whereas in our model the follower has the greater decision-pertinent knowledge and discloses it strategically. To the best of our knowledge, there has been very limited work studying \emph{follower} deception in Stackelberg games.
The most relevant to ours is perhaps \cite{xu2016signaling}, which studies signaling in Bayesian Stackelberg games and also considers the possibility that a follower may strategically misreport his type.  However, the question they study is completely different from ours --- they design signaling schemes for different follower types whereas we design the leader's mixed strategy for each follower type. The contrasting complexity  between their model (polynomial-time solvable for normal form games) and ours (NP-hard) also highlights the intrinsic difference.
Even more importantly, though, in our model type reporting takes the conceptual form of a behaviour inducing message. The common information asymmetry (see e.g.~\cite{rjjx15,xrtd15,dkq16}) is reversed between the leader and the follower.

\section{A Motivating Example}\label{sec:example}

\newcommand{\fone}{\mathsf{1^x}}
\newcommand{\ftwo}{\mathsf{2^x}}
\newcommand{\fthr}{\mathsf{3^x}}
\newcommand{\lone}{\mathsf{1}}
\newcommand{\ltwo}{\mathsf{2}}

We  illustrate how the follower's deceptive behavior may result in highly suboptimal leader strategy, and how we can overcome this issue using carefully-designed leader strategies. For illustrative purpose, we use a real-world security game example, though the underlying phenomenon also generalizes to other applications.

A security game is a Stackelberg game played between a \emph{defender} (the leader) and an \emph{attacker} (the follower). Concretely, we consider a security agency who wants to protect  two conservation areas, i.e., areas $1$ and $2$, from a \emph{poacher}'s attack. The defender can only patrol one of these two areas, while the poacher chooses one area to attack.
At different periods of the year, the poacher's payoffs change due to price fluctuation of wildlife products on black markets. We capture this uncertainty using two possible types of poacher payoffs, i.e.,  \emph{type} $A$ and $B$. Assume that the defender's payoff is \emph{independent} of the poacher's type.  The poacher knows his true type,  but the defender only has a prior belief that each type shows up with probability $0.5$. (This can be obtained via,  e.g., surveying past prices on black markets.) The payoff matrices of the defender (row player) and the two poacher types (column player) are shown below.
\begin{figure}[h!]
\renewcommand{\arraystretch}{1.3}
\newcolumntype{R}{>{\rule[-0.5em]{0pt}{1.5em}$}r<{$}}
\newcolumntype{L}{>{\rule[-0.5em]{0pt}{1.5em}\raggedleft\arraybackslash$}p{5mm}<{$}}
\newcolumntype{C}{>{\rule[-0.5em]{0pt}{1.5em}\centering\arraybackslash$}p{5mm}<{$}}
\hspace{-3mm}
\centering
\begin{tabular}{R|L|L|}
\multicolumn{1}{R}{}&	\multicolumn{1}{C}{\fone}	&  \multicolumn{1}{C}{\ftwo}  \\
\cline{2-3}
 \lone	&	1	& -1	 \\\cline{2-3}
 \ltwo	&	-1	& 0.99	 \\\cline{2-3}
\multicolumn{1}{R}{}& \multicolumn{2}{C}{\rule[-0.2em]{0pt}{2em} \emph{defender}} \\
\end{tabular}
\hspace{10mm}
\begin{tabular}{R|L|L|}
\multicolumn{1}{R}{}&	\multicolumn{1}{C}{\fone}	&  \multicolumn{1}{C}{\ftwo}  \\
\cline{2-3}
 \lone	&	-1	& 1/3  \\\cline{2-3}
 \ltwo	&	3	& -1	 \\\cline{2-3}
\multicolumn{1}{R}{}& \multicolumn{2}{C}{\rule[-0.2em]{0pt}{2em} \hspace{-4mm}\emph{poacher type A}} \\
\end{tabular}
\hspace{10mm}
\begin{tabular}{R|L|L|}
\multicolumn{1}{R}{}&	\multicolumn{1}{C}{\fone}	&  \multicolumn{1}{C}{\ftwo}  \\\cline{2-3}
 \lone	&	-1	& 1	 \\\cline{2-3}
 \ltwo	&	1	& -1	 \\\cline{2-3}
\multicolumn{1}{R}{}& \multicolumn{2}{C}{\rule[-0.2em]{0pt}{2em} \hspace{-4mm}\emph{poacher type B}} \\
\end{tabular}
\end{figure}

The only difference between these poacher types is their value for successful attacks, which is affected by wildlife product prices at that time period. Standard calculation shows that if the poacher has type $A$, the optimal defender strategy is $(3/4, 1/4)$, i.e., patrolling areas $1$  and $2$ with probability $3/4$ and $1/4$, respectively, resulting in poacher utility $0$ at both areas. By the standard assumption, the poacher thus breaks the tie \emph{in favor} of the defender and attacks area~$1$,\footnote{This assumption is without loss of generality. The corresponding solution concept, the \emph{strong Stackelberg equilibrium}, remains the most widely adopted solution concept in the literature of Stackelberg games (see e.g., \cite{osborne1994course,conitzer2006computing,tambe2011security}). Normally, via infinitesimal strategy variations, the leader can induce the follower to play \emph{any} best response action to the leader's benefit. For example, the defender can play $(\frac{3}{4}-\epsilon, \frac{1}{4}+\epsilon)$ with $\epsilon\to 0$, so that the poacher will strictly prefer to attack area 1 while the change to the defender's utility can be arbitrarily small. Our empirical evaluation in Section~\ref{sec:experiment} will further confirm the robustness of our solutions against the tie-breaking issue.\label{fn:tie-breaking}}
yielding defender utility $1/2$. Similarly, the optimal defender strategy is $(1/2,1/2)$ against a  type-$B$ poacher, which induces the poacher to also attack area $1$ and yields defender utility $0$.

Ideally, the defender would like to play the optimal patrolling strategy against the poacher's true type at any time which, however, is unknown to the defender. Nevertheless, if the poacher were to behave truthfully according to his type, the defender can easily learn the poacher's type, e.g., by observing their (different) best responses to strategy $(3/5,2/5)$. Indeed,  previous approaches for learning the optimal leader strategy rely crucially on the  assumption that the follower will truthfully respond (e.g., \cite{letchford2009learning,blum2014learning,haghtalab2016three,roth2016watch}). However, in this example, a type-$A$ poacher has strong incentives to be untruthful: by imitating a type-$B$ poacher and reacting exactly according to the type-$B$ payoff structure, a type-$A$ poacher is able to induce the defender to play $(1/2,1/2)$ --- the optimal strategy against a type-$B$ poacher --- which results in true utility ${\frac{1}{2} \cdot (-1)} + {\frac{1}{2} \cdot (3)} = 1$ for the type-$A$ poacher. This strictly improves upon his previous utility $0$ of reacting truthfully.
Simple calculation shows that a type-$B$ poacher does not benefit from  deception and will behave truthfully. As a result, the defender will not be able to distinguish type $A$ from $B$ under the poacher's strategic deception.

The issue raised above is due to the fact that the defender focuses solely on optimizing her utility without considering the poacher's strategic behavior. We now show how the defender can overcome this issue by taking into account the poacher's deception. It happens that in this  example the optimal policy for the defender  is to still play $(3/4, 1/4)$ for type $A$ and $(1/2,1/2)$ for type $B$, but induce a type-$B$ poacher to \emph{break tie ``against'' the defender} by attacking area $2$.\footnote{This tie breaking can be induced explicitly by playing $(\frac{1}{2}+\epsilon,\frac{1}{2}-\epsilon)$ for an infinitesimal $\epsilon$.}
This  slightly decreases the defender's utility against a type-$B$ poacher (from $0$ to $\frac{1}{2}\cdot(-1) + \frac{1}{2}\cdot0.99 = -0.005$), but it erases the incentive of a type-$A$ poacher to imitate type $B$: if the type-$A$ poacher imitates $B$ to attack area $2$ under defender strategy $(1/2,1/2)$, his expected utility would become $-1$, which is worse than his utility $0$ of behaving truthfully.

\paragraph{Remarks}
(i) This example illustrates an intriguing phenomenon: when there is follower deception, it may be undesirable to always induce tie breaking in favor of the leader, since other tie breaking strategies may enforce more desirable follower behaviors.
If the defender adopts this nuanced strategy, the poacher will have no incentive to deceive. Thus, from the perspective of machine learning, what we are  designing can be viewed as the best strategy the leader can learn in the presence of follower deception.

(ii) One might wonder why the defender does not ignore the poacher's deception by simply playing a single strategy against all types of poachers, which is precisely the optimal commitment in a Bayesian Stackelberg game. This reason is that our more sophisticated approach can achieve better defender utility. In this example, it can be computed that the defender would play $(1/2, 1/2)$ in the Bayesian Stackelberg equilibrium, obtaining utility $0$. However, the strategy we designed above yields defender utility $\frac{1}{4} - \frac{1}{400}$.
A more formal result regarding the utility improvement will be presented in Proposition~\ref{prp:icexists}.

\begin{figure}[t]
\renewcommand{\arraystretch}{1.3}
\newcolumntype{R}{>{\rule[-0.5em]{0pt}{1.5em}$}r<{$}}
\newcolumntype{L}{>{\rule[-0.5em]{0pt}{1.5em}\raggedleft\arraybackslash$}p{5mm}<{$}}
\newcolumntype{C}{>{\rule[-0.5em]{0pt}{1.5em}\centering\arraybackslash$}p{5mm}<{$}}
\hspace{-3mm}
\centering
\begin{tabular}{R|L|L|}
\multicolumn{1}{R}{}&	\multicolumn{1}{C}{\fone}	&  \multicolumn{1}{C}{\ftwo}  \\
\cline{2-3}
 \lone	&	1	& 0	 \\\cline{2-3}
 \ltwo	&	0	& \epsilon	 \\\cline{2-3}
\multicolumn{1}{R}{}& \multicolumn{2}{C}{\rule[-0.2em]{0pt}{2em} \emph{leader}} \\
\end{tabular}
\hspace{10mm}
\begin{tabular}{R|L|L|}
\multicolumn{1}{R}{}&	\multicolumn{1}{C}{\fone}	&  \multicolumn{1}{C}{\ftwo}  \\
\cline{2-3}
 \lone	&	0	& 0.2  \\\cline{2-3}
 \ltwo	&	0.8	& 1	 \\\cline{2-3}
\multicolumn{1}{R}{}& \multicolumn{2}{C}{\rule[-0.2em]{0pt}{2em} \hspace{-4mm}\emph{follower type A}} \\
\end{tabular}
\hspace{10mm}
\begin{tabular}{R|L|L|}
\multicolumn{1}{R}{}&	\multicolumn{1}{C}{\fone}	&  \multicolumn{1}{C}{\ftwo}  \\\cline{2-3}
 \lone	&	0.2	& 0	 \\\cline{2-3}
 \ltwo	&	1	& 0.8	 \\\cline{2-3}
\multicolumn{1}{R}{}& \multicolumn{2}{C}{\rule[-0.2em]{0pt}{2em} \hspace{-4mm}\emph{follower type B}} \\
\end{tabular}
\caption{The price of ignoring follower deception can be huge. Let both follower types appear with probability $0.5$. When deception is ignored, the leader assumes the follower reports truthfully and will play $(0,1)$ for type $A$ and $(1,0)$ for type $B$. However, this gives a type-$B$ follower an incentive to misreport that he has type $A$, yielding leader utility $\epsilon$ overall. Now, if the leader properly deals with follower deception and plays, instead, $(1,0)$ for both types $A$ and $B$, a type-$B$ follower will have no incentive to misreport; the leader is able to obtain utility $0.5$ in this case. The ratio $0.5/\epsilon$ is unbounded when $\epsilon$ is arbitrarily close to $0$. \label{fig:price}}
\end{figure}

\begin{figure}[t]
\renewcommand{\arraystretch}{1.5}
\centering
\begin{tabular}{l|r|r|}
\multicolumn{1}{l}{}&	\multicolumn{1}{c}{${C}$}	&  \multicolumn{1}{c}{${D}$} \\\cline{2-3}
\rule[-0.5em]{0pt}{1.5em} ${C}$	&	$-1$, \tikz[baseline]{\node[draw=red, semithick,fill=white!20,anchor=base,ellipse,inner xsep=-1pt,inner ysep=0pt] (d1) {$-1$};}	&	$-3$, $\phantom{-}0$ \\\cline{2-3}
\rule[-0.5em]{0pt}{1.5em} ${D}$	&	$0$, $-3$	&	$-2$, $-2$ \\\cline{2-3}
\end{tabular}
\begin{tikzpicture}[overlay,remember picture]
\draw[red,semithick,->] (d1) to [in=185,out=70] +(50:7mm) node[right]{\hspace{-0.7mm}\raisebox{0.5mm}{$1$}};
\end{tikzpicture}
\vspace{2mm}
\caption{Follower deception is not always bad. If the follower (column player) deceives the leader (row player) into believing that he has utility $1$, instead of $-1$, for the  action profile $({C},{C})$ (annotated in red), the leader would feel reassured to commit to ${C}$, resulting in utility $-1$ for both players. \label{fig:beneficial_decept}}
\end{figure}

(iii) The price of ignoring follower deception can be huge. Figure~\ref{fig:price} illustrates an example in which the price --- the ratio between the leader's utilities when deception is ignored and when deception is properly dealt with --- is unbounded. (Payoffs are normalized to be in $[0,1]$ for the ratio to be meaningful.)
Nevertheless, perhaps counter-intuitively, follower deception is not always bad for the leader. One example is a Stackelberg game version of the prisoner's dilemma illustrated in Figure~\ref{fig:beneficial_decept}.

\section{The Model}
\label{sec:model}

\subsection{Stackelberg Game Basics}

A Stackelberg game is played between a \emph{leader} and a \emph{follower}.
A normal-form Stackelberg game is given by two matrices  $u^\xl,u^\xf \in \mathbb{R}^{m \times n}$, which are the payoff matrices of the leader (row player) and follower (column player), respectively. We use  $u^\xl(i,j)$ to denote a generic entry of  $u^\xl$, and use  $[m] = \{ 1,...,m\}$ to denote the set of the leader's pure strategies (also called \emph{actions}).
A mixed strategy of the leader is a probabilistic distribution over $[m]$, denoted by a vector $\mathbf{x} \in \Delta_m = \{ \mathbf{p} \geq 0 : \sum_{i \in [m]} p_i = 1 \}$. With slight abuse of notation, we denote by $u^\xl(\mathbf{x},j)$ the expected utility $\sum_{i \in [m]} x_i \cdot u^\xl(i,j)$.
Notation for the follower is the same by changing the labels. We will sometimes write a pure strategy $i$ in positions where a mixed strategy is expected, in which case it represents the same strategy in its mixed strategy form, i.e., a basis vector $\mathbf{e}_{i}$.

In a Stackelberg game, the leader moves first by committing to a mixed strategy $\mathbf{x}$, and the follower then best responds to $\mathbf{x}$.
Without loss of generality, it is assumed that the follower's best response is a pure strategy with ties broken in favor of the leader if there are multiple best responses.\footnote{See footnote~\ref{fn:tie-breaking}.}
Under this assumption, the leader mixed strategy that maximizes her expected utility leads to a \emph{strong Stackelberg equilibrium} (SSE), which is the standard solution concept of Stackelberg games. Formally, denote by $\BR(\mathbf{x}):=\arg\max_{j\in[n]} u^\xf( \mathbf{x}, j)$ the set of follower best responses to a leader strategy $\mathbf{x}$; a pair of strategies $\mathbf{x}^*$ and $j^*$ forms an SSE if
\begin{align*}
	\langle \mathbf{x}^*, j^* \rangle \in \arg\max\nolimits_{\mathbf{x} \in \Delta_m, j \in \BR(\mathbf{x}) } u^\xl (\mathbf{x}, j ). \label{eq:sse}
\end{align*}

\subsection{Stackelberg Games with Imitative Follower Deception}

We now describe a basic model that captures (imitative) follower deception; the reader may refer back to our motivating example as an instantiation of this model.
We consider a leader who faces a follower with an uncertain type, which falls in a discrete set $\Theta$.  Each type $\theta \in \Theta$ corresponds to a different follower payoff matrix $u_\theta^\xf \in \mathbb{R}^{m\times n}$.  For ease of presentation, we will assume that the leader's payoff does \emph{not} depend on $\theta$, though we remark that all our results easily generalize to the case with type-dependent leader payoffs.
To model the leader's prior knowledge regarding the follower's type, we adopt the classic Bayesian perspective and assume  that the leader has  a prior distribution $\pi$, i.e., each type $\theta$ appears with probability $\pi_\theta$. Note that $\pi, \Theta$ and the payoff matrices are common knowledge.

We assume that the leader can commit to a \emph{policy} --- a ``menu'' that specifies a mixed strategy $\mathbf{x}^{\theta}$ to be played against each follower type $\theta$.
In addition, we assume that a follower best response $j^\theta \in \BR_{{\theta}}(\mathbf{x}^\theta)$ is also specified for each $\theta$ in the policy, where $\BR_{{\theta}}(\mathbf{x}) := \arg\max_{j} u_{{\theta}}^\xf (\mathbf{x},j)$ now denotes the {best response set} of a type-$\theta$ follower against leader strategy $\mathbf{x}$; hence, in the case of a tie, the follower will be induced to play this specific response.
Note that, same as in the motivating example, $j^\theta$ is not necessarily the one that maximizes the utility the leader obtains on follower type $\theta$; rather, the leader would prefer to induce the follower to a carefully chosen follower action that discourages follower deception to some extent and benefits the leader's overall utility.

After observing the leader's policy, a follower of type $\theta$ reports a type $\hat{\theta}$ to the leader, and the leader then plays the mixed strategy $\mathbf{x}^{\hat{\theta}}$ as specified by her policy. The follower then ``best'' responds to $\mathbf{x}^{\hat{\theta}}$ with $j^{\hat{\theta}}$, as if his type is $\hat{\theta}$. Naturally, the follower will report a type that maximizes his actual expected utility, i.e., $\hat{\theta} = \arg\max_{\beta \in \Theta}  u_\theta^\xf(\mathbf{x}^{\beta}, j^{\beta})$, resulting in the leader to obtain utility $u^\xl(\mathbf{x}^{\hat{\theta}},j^{\hat{\theta}})$.
The reporting step in our model is straightforward in various applications (see Section~\ref{sec:intro}), whereas in some other applications, reporting abstracts the process that the leader learns the follower's type by interacting with the follower (as in the motivating example).

\smallskip
In summary, the game proceeds as follows:
\begin{itemize}
\item[1.]
The leader commits to a policy $\sigma = \{ \otc^{\theta} \}_{\theta \in \Theta}$ that prescribes, for each reported follower type $\theta \in \Theta$, an \emph{outcome} $o^\theta = \langle \mathbf{x}^\theta, j^\theta \rangle$ such that $j^{\theta}  \in \BR_{{\theta}}(\mathbf{x}) := \arg\max_{j} u_{{\theta}}^\xf (\mathbf{x},j)$.

\item[2.]
After observing the leader policy $\sigma$, a follower of type $\theta$, who appears with probability $\pi_\theta$, reports a best type $\hat{\theta}(\sigma) = \arg\max_{\beta \in \Theta}  u_\theta^\xf(\mathbf{x}^{\beta}, j^{\beta})$ that maximizes his actual utility; same as the standard SSE assumption, we assume the follower breaks ties by reporting a type in favor of the leader. This results in overall leader utility
\begin{align*}
     U^\xl(\sigma) :=  \sum_{\theta \in \Theta } \pi_{\theta} \cdot u^\xl ( \mathbf{x}^{\hat{\theta}(\sigma)}, j^{\hat{\theta}(\sigma)} ).
\end{align*}
\end{itemize}

The computational task we examine is $\max_{\sigma} U^\xl(\sigma)$. For convenience, we will also write $u_\theta^\xf (\otc) = u_\theta^\xf (\mathbf{x}, j)$ and $u^\xl (\otc) = u^\xl (\mathbf{x}, j)$ for an outcome $o=\langle \mathbf{x}, j \rangle$.

\paragraph{Remarks}
Our model can be viewed as a generalization of the classic Bayesian Stackelberg game (BSG) model where the leader's policy is restricted to prescribe the same mixed strategy for all follower types \cite{conitzer2006computing}.
Yet, our model cannot be reduced to a classic BSG, a notable difference being that the leader's strategy (policy) space in our model is \emph{non-convex} due to the extra component $j^\theta$ in the prescribed outcomes whereas it is convex for a BSG.

\subsection{Incentive Compatibility (IC)}
One can impose in addition the IC constraints on $\sigma$ so that reporting truthfully is a (weakly) dominant strategy for every follower type, i.e., $ u_\theta^\xf(\mathbf{x}^{\theta}, j^{\theta}) \geq u_\theta^\xf(\mathbf{x}^{\beta}, j^{\beta}) $ for any $\beta \in \Theta$. We will consider model variants both with and without IC constraints.
The following result shows that an IC policy always exists, and even when IC constraints are imposed, our approach always achieves at least the leader utility in a Bayesian Stackelberg equilibrium. This echoes our observation in the motivating example.

\begin{proposition}
\label{prp:icexists}
There always exists an IC policy, and the optimal IC policy achieves at least the leader utility in a Bayesian Stackelberg equilibrium.
\end{proposition}
\begin{proof}
Let $\mathbf{x}$ be the leader strategy in a Bayesian Stackelberg equilibrium (BSE), and let $j^\theta$ be the best response of a type-$\theta$ follower in the BSE for each $\theta$, i.e., $j^\theta \in \arg\max_{j\in \BR_{\theta}(\mathbf{x}) } u^\xl (\mathbf{x},j)$.
A policy that prescribes $o^\theta = \langle \mathbf{x} , j^\theta \rangle$ for each follower type $\theta \in \Theta$ is trivially IC because by reporting truthfully the follower will also be induced to play his \emph{true} best response $j^\theta$ to $\mathbf{x}$.
In addition, it offers the leader utility $\sum_{\theta \in \Theta} \pi_\theta \cdot u^\xl (\mathbf{x},j^\theta)$ --- exactly what she obtains in the BSE.
\end{proof}

\section{Computing the Optimal Policy}
\label{sec:complexity}

\renewcommand{\fone}{\mathsf{1^F}}
\renewcommand{\ftwo}{\mathsf{2^F}}
\renewcommand{\fthr}{\mathsf{3^F}}

\newcommand{\opt}{\mbox{\textnormal{\textsc{Opt-Policy}}}}
\newcommand{\optic}{\mbox{\textnormal{\textsc{Opt-Policy-IC}}}}
\newcommand{\optx}{\mbox{\textnormal{\textsc{Opt-XPolicy}}}}
\newcommand{\optxic}{\mbox{\textnormal{\textsc{Opt-XPolicy-IC}}}}

In this section, we study the complexity and algorithms for computing the optimal leader policy. We will refer to the problem of computing the optimal policy {\opt} and, when IC constraints are imposed, {\optic}.

\subsection{Hardness of Approximation}

We show in Theorems~\ref{thm:apx_optBay} and \ref{thm:apx_icBay} that it is NP-hard even just to approximate the optimal policy to within a meaningful ratio, with or without IC constraints, and even when the follower has only a small number of actions.
The bounds in the inapproximability results are essentially tight by Theorem~\ref{thm:1_Theta_apx}.
In contrast, there is an efficient algorithm for a small follower type set, which we show in Theorem~\ref{thm:opt_fixTheta} after we present an MILP (mixed integer linear program)-based algorithm for the general case.
Our approximation ratios are all \emph{multiplicative}; by convention, we normalize the leader's payoff parameters to be in $[0,1]$ in order for the ratios to be meaningful.

\begin{theorem}
\label{thm:apx_optBay}
For any constant $\epsilon > 0$, {\opt} does not admit any polynomial-time $\frac{1} {{(|\Theta|-1)}^{1-\epsilon}}$-approximation algorithm unless P = NP, even when the follower has only three actions.
\end{theorem}

\begin{proof}
We show a reduction from the \textsc{Max-Independent-Set} problem, which asks for the size of the maximum independent set in a graph $G = (V,E)$.
A set of nodes $V' \subseteq V$ is an \emph{independent set} of $G$ if no edge in $E$ connects any pair of nodes in $V'$; an independent set is maximum if there exists no other independent set with a larger size.
It is known that no polynomial-time $\frac{1} {{|V|}^{1-\epsilon}}$-approximation algorithm exists for \textsc{Max-Independent-Set}, unless P = NP \cite{zuckerman2006linear}.

We construct a Stackelberg game with $|V|+1$ follower types $\Theta = \{ \theta_* \} \cup \{ \theta_v : v \in V \}$, where each $\theta_v$ corresponds to node $v$.
Let $\pi_{\theta_*} = 0$ and $\pi_{\theta} = \frac{1}{|\Theta|-1}$ for all other $\theta$.
The leader has $2|V|+1$ actions $\{a_v\!: v \!\in\! V \} \cup \{ b_v \!: v \!\in\! V\} \cup \{a_0\}$; and the follower has
three actions $\{ \fone, \ftwo, \fthr \}$.
The payoffs are given in Figure~\ref{fig:thm:apx_optBay}, where $\mathcal{N}(v)$ denotes the set of neighbouring nodes of $v$, and all the empty entries are $0$.

\begin{figure}[t]
\renewcommand{\arraystretch}{1.2}
\newcolumntype{R}{>{$}r<{$}}
\newcolumntype{C}{>{$}c<{$}}
\centering
	\begin{tabular}{r|C|C|C|}
		\multicolumn{1}{R}{}&	\multicolumn{1}{C}{\fone}	&  \multicolumn{1}{C}{\ftwo}  &  \multicolumn{1}{C}{\fthr} \\\cline{2-4}
		\rule[-0.5em]{0pt}{1.5em} $a_0$	&	0.5	& 1	    & 1     \\\cline{2-4}
        \rule[-0.5em]{0pt}{1.5em} $a_v$	&		& 0.5	& 0.5   \\\cline{2-4}
		\rule[-0.5em]{0pt}{1.5em} $b_v$	&		& 1	    & 1     \\\cline{2-4}
		\rule[-0.5em]{0pt}{1.5em} $a_{v'}: v' \in \mathcal{N}(v)$	
                                        &		& 1	    & 1     \\\cline{2-4}
		\rule[-0.5em]{0pt}{1.5em} \emph{otherwise}	
                                        &		& 0.5	& 1     \\\cline{2-4}
		\multicolumn{1}{r}{} & \multicolumn{3}{c}{\rule[-0.2em]{0pt}{2em} \!\!\!\emph{follower type} $\theta_v$\!\!\!} \\
	\end{tabular}
\hspace{4mm}
	\begin{tabular}{r|C|C|C|}
        \multicolumn{1}{R}{}&	\multicolumn{1}{C}{\fone}	&  \multicolumn{1}{C}{\ftwo}  &  \multicolumn{1}{C}{\fthr} \\\cline{2-4}
		\rule[-0.5em]{0pt}{1.5em} \emph{any} $i$ &	1	& 	    &      \\\cline{2-4}
        \multicolumn{4}{r}{\rule[-0.2em]{0pt}{2em}  \emph{follower type} $\theta_*\!\!\!\!$} \\
	\end{tabular}
\hspace{4mm}
	\begin{tabular}{r|C|C|C|}
		\multicolumn{1}{R}{}&	\multicolumn{1}{C}{\fone}	&  \multicolumn{1}{C}{\ftwo}  &  \multicolumn{1}{C}{\fthr} \\\cline{2-4}
		\rule[-0.5em]{0pt}{1.5em} \emph{any} $i$ &	1	& 	    &      \\\cline{2-4}
        \multicolumn{1}{r}{} & \multicolumn{3}{c}{\rule[-0.2em]{0pt}{2em} \emph{leader}} \\
	\end{tabular}
\caption{Payoffs matrices (proof of Theorem~\ref{thm:apx_optBay}). \label{fig:thm:apx_optBay}}
\end{figure}

The following are a few observations about the game:
\begin{itemize}
\item
For follower type $\theta_*$, $j = \fone$ strictly dominates all the other actions.
A follower can only be induced to play $j = \fone$ when $\theta_*$ is reported, in which case the leader gets utility $1$.

\item
For each $\theta_v$, when they are reported, $j=\fone$ cannot be induced by any leader strategy because it is strictly dominated by $j = \ftwo$.
As a result, the leader gets $0$ once $\theta_v$ is reported.

\item
Given the above, the leader's overall utility is proportional to the number of follower types $\theta_v$ that are motivated to report $\theta_*$.
\end{itemize}

Now for any independent set $V'$ of size $k$ in $G$, we can construct the following leader policy $\sigma$:
$\otc^{\theta_*} = \langle \vecone{a_0}, \fone \rangle$,
$\otc^{\theta_v} = \langle \vecone{a_v}, \ftwo \rangle$ if $v \in V'$,
and $\otc^{\theta_v} = \langle \vecone{b_v}, \ftwo \rangle$ if $v \notin V'$. Given this policy, all types $\theta_v$, $v \in V'$, will find reporting $\theta_*$ optimal, while all the others are better-off reporting truthfully.
As a result, the leader receives overall utility $\frac{k}{|\Theta|-1}$.

Conversely, we claim that for any policy $\sigma$, the set of nodes $v$ of which the corresponding types $\theta_v$ are incentivized to report $\theta_*$ form an independent set, i.e., $\{v\in V: \hat{\theta}_v(\sigma) = \theta_*\}$ is an independent set.
To see this, suppose for a contradiction that types $\theta_v$ and $\theta_{v'}$ are incentivized to report $\theta_*$ but $(v,v') \in E$. By our observations above, both $\theta_v$ and $\theta_{v'}$ will be induced to play $\fone$ and obtain utility $0.5$ each.
Thus, we need the following in order for them to have no incentive to report another type.
\begin{itemize}
\item[(i)]
$x_{a_{v'}}^{\theta_{v'}} =0$, since otherwise the fourth row of the payoff matrix of $\theta_v$ would be chosen with some probability, and $\theta_v$ would be better-off reporting $\theta_{v'}$ to obtain utility strictly greater than $0.5$.

\item[(ii)]
$x_{a_0}^{\theta_{v'}} = x_{b_{v'}}^{\theta_{v'}} =0$, and $x_{a_{v''}}^{\theta_{v'}} =0$ for all $v'' \in \mathcal{N}(v')$, since otherwise one of the first, third and fourth rows of the payoff matrix of $\theta_{v'}$ would be chosen with some probability, and $\theta_{v'}$ would be better-off reporting truthfully to obtain utility strictly greater than $0.5$.
\end{itemize}
It follows that, now $\mathbf{x}^{\theta_{v'}}$ will pick the last row of the payoff matrix with probability $1$, so $\fthr$ becomes the only best response of a type-$\theta_{v'}$ player against $\mathbf{x}^{\theta_{v'}}$.
Thus, we have $o^{\theta_{v'}} = \langle \mathbf{x}^{\theta_{v'}}, \fthr \rangle$, in which case, however, $\theta_{v'}$ is able to obtain utility $1$ by reporting truthfully, which contradicts the assumption that $\theta_{v'}$ is incentivized to report $\theta^*$.

Therefore, the number of follower types that are motivated to report $\theta_*$ by the optimal policy  is exactly to the size of the maximum independent set (and proportional to the overall leader utility).
Any $\frac{1} {{(|\Theta|-1)}^{1-\epsilon}}$-approximation algorithm would also provide a $\frac{1} {{|V|}^{1-\epsilon}}$-approximation to the \textsc{Max-Independent-Set} problem, which exists unless P = NP.
This completes the proof.
\end{proof}

Next, we show that the same inapproximability result for {\optic}. Our reduction will still be from \textsc{Max-Independent-Set} but the underlying idea and the instance constructed are different: the previous reduction maps follower types that \emph{violate} IC to an independent set, which cannot be applied when IC constraints are imposed.

\begin{theorem}
\label{thm:apx_icBay}
  For any constant $\epsilon > 0$, {\optic} does not admit any polynomial-time $\frac{1} {{|\Theta|}^{1-\epsilon}}$-approximation algorithm unless P = NP, even when the follower has only three actions.
\end{theorem}
\begin{proof}
    We show a reduction from the \textsc{Max-Independent-Set} problem.
   Given an instance of \textsc{Max-Independent-Set} specified by a graph $G = (V,E)$, we construct a game with $|V|$ follower types $\Theta = \{ \theta_v : v \in V \}$; each $\theta_v$ corresponds to a node $v$ and appears with probability $\frac{1}{|\Theta|}$.
    The leader has $2|V|$ actions $\{a_v : v \in V \} \cup \{b_v: v \in V\}$; and the follower has
  three actions $\{ \fone, \ftwo, \fthr \}$.
  The payoffs are given in Figure~\ref{fig:thm:apx_icBay}, where $\mathcal{N}(v)$ denotes the set of neighbouring nodes of $v$, and all the empty entries are $0$.

\begin{figure}[t]
\renewcommand{\arraystretch}{1.2}
\newcolumntype{R}{>{$}r<{$}}
\newcolumntype{C}{>{$}c<{$}}
\centering
	\begin{tabular}{r|C|C|C|}
		\multicolumn{1}{R}{}&	\multicolumn{1}{C}{\fone}	&  \multicolumn{1}{C}{\ftwo}  &  \multicolumn{1}{C}{\fthr} \\\cline{2-4}
		\rule[-0.5em]{0pt}{1.5em} $a_v$	&		& \phantom{0.0}	& \phantom{0.0} \\\cline{2-4}
		\rule[-0.5em]{0pt}{1.5em} $b_v$	&		& 1	& 1 \\\cline{2-4}
		\rule[-0.5em]{0pt}{1.5em} $a_{v'}: v'\in \mathcal{N}(v)$	&	0.5	& 	&  1 \\\cline{2-4}
		\rule[-0.5em]{0pt}{1.5em} \emph{otherwise}	&		& 	&  1 \\\cline{2-4}
		\multicolumn{4}{r}{\rule[-0.2em]{0pt}{2em} \emph{follower type} $\theta_v\!\!$} \\
	\end{tabular}
	\hspace{5mm}
	\begin{tabular}{r|C|C|C|}
		\multicolumn{1}{R}{}&	\multicolumn{1}{C}{\fone}	&  \multicolumn{1}{C}{\ftwo}  &  \multicolumn{1}{C}{\fthr} \\\cline{2-4}
		\rule[-0.5em]{0pt}{1.5em} \emph{any} $i$ &	1	& 	    &      \\\cline{2-4}
        \multicolumn{1}{r}{} & \multicolumn{3}{c}{\rule[-0.2em]{0pt}{2em} \emph{leader}}
	\end{tabular}
\caption{Payoffs matrices (proof of Theorem~\ref{thm:apx_icBay}). \label{fig:thm:apx_icBay}}
\end{figure}

Now for any independent set $V'$ of size $k$ in $G$, consider a leader policy $\sigma$ with:
$\otc^{\theta_v} = \langle \vecone{a_v}, \fone \rangle$ if $v \in V'$,
and $\otc^{\theta_v} = \langle \vecone{b_v}, \ftwo \rangle$ if $v \notin V'$.
In this policy, all types find truthful report to be optimal. Hence, the policy is IC and offers the leader overall utility $\frac{k}{|\Theta|}$.

Conversely, for any IC policy $\sigma$, we claim that the set of nodes $v$ of which the corresponding follower types $\theta_v$ are motivated to respond with $j=\fone$ form an independent set, i.e., $\{v\in V: o^{\theta_v} = \langle i, \fone \rangle \text{ for some } i\}$ is an independent set.
Suppose for a contradiction that $\theta_v$ and $\theta_{v'}$ both respond with $j=\fone$ while $(v,v') \in E$.
Observe that, when $\theta_v$ is reported, $j=\fone$ can be induced only by leader strategy $\vecone{a_v}$ because it is strictly dominated by $j = \fthr$ in all other cases.
Thus, $\otc^{\theta_v} = \langle \vecone{a_v}, \fone \rangle$ and $\otc^{\theta_{v'}} = \langle \vecone{a_{v'}}, \fone \rangle$, in which case both $\theta_v$ and $\theta_{v'}$ obtain utility $0$, so $\theta_{v'}$ would be better-off reporting $\theta_v$, which contradicts the assumption that $\sigma$ is IC.

Therefore, under the optimal IC policy, the number of follower types that respond with $j = \fone$ (which is proportional to the leader's overall utility) is equal to the size of the maximum independent set in $G$. Any $\frac{1} {{|\Theta|}^{1-\epsilon}}$-approximation algorithm would also provide a $\frac{1} {{|V|}^{1-\epsilon}}$-approximation to the \textsc{Max-Independent-Set} problem.
\end{proof}

The bounds in Theorems \ref{thm:apx_optBay} and \ref{thm:apx_icBay} are essentially tight as suggested by a polynomial-time approximation algorithm in the next theorem with the matching approximation ratio.

\begin{theorem}
\label{thm:1_Theta_apx}
Both {\opt} and {\optic} admit a polynomial-time $\frac{1}{|\Theta|}$-approximation algorithm.
\end{theorem}

\begin{proof}
We show an approximation algorithm for {\opt}. The algorithm for {\optic} follows a similar procedure.

For each follower type $\theta \in \Theta$, we compute a policy $\sigma_\theta$ that maximizes the actual utility the leader obtains on this specific type.
We show that (i) $\sigma_\theta$ can be computed in polynomial time, and (ii) the best $\sigma_\theta$ over all $\theta \in \Theta$ achieves at least $\frac{1}{|\Theta|}$ of the leader utility under the optimal policy.

\paragraph{(i) Polynomial-time computability.} To compute $\sigma_\theta$, the idea is to enumerate all possible reporting strategies $\beta\in \Theta$ of type $\theta$ and, for each $\beta$, compute the best policy \emph{under the constraint that it is indeed optimal for a type-$\theta$ follower to report $\beta$}.
An observation that helps simplifying the computation is that we only need to guarantee reporting $\beta$ to be (weakly) preferred by a type-$\theta$ follower to reporting $\theta$, i.e., the truthful report, so it suffices to consider only the two associated outcomes in the leader's policy, i.e., $o^\theta = \langle \mathbf{x}^\theta, j^\theta \rangle$ and $o^{\beta} = \langle \mathbf{x}^{\beta}, j^{\beta} \rangle$.
This is due to the fact that a policy that prescribes the same strategy $\mathbf{x}^\theta$ for all other reports $\beta' \notin \{\theta,\beta\}$ will trivially make the truthful report (weakly) preferred to all other reporting strategies $\beta' \notin \{\theta,\beta\}$: by reporting truthfully, the follower will already be induced to play his true best response to $\mathbf{x}^\theta$.
Now, the problem reduces to solving the following LP (liner program), in which $\mathbf{x}^\theta$ and $\mathbf{x}^{\beta}$ are the variables, for every pair of possible values of $j^\theta$ and $j^{\beta}$. Since $j^\theta, j^{\beta} \in [n]$, there are $n^2$ LPs to be solved; the polynomial-time computability follows immediately.
\begin{subequations}\label{eq:op1}
\begin{talign}
  \text{maximize} \quad &u^\xl(\mathbf{x}^{\beta}, j^{\beta}) \tag{\ref{eq:op1}} \\
  \text{subject to} \quad & %
  u_\theta^\xf (\mathbf{x}^\theta, j^\theta) \ge u_\theta^\xf (\mathbf{x}^\theta, j) && \text{for } j \in [n]  \label{eq:op1_a} \\
  &u_{\beta}^\xf (\mathbf{x}^{\beta}, j^{\beta}) \ge u_{\beta}^\xf (\mathbf{x}^{\beta}, j) && \text{for } j \in [n] \label{eq:op1_b}\\
  &u_\theta^\xf(\mathbf{x}^{\beta}, j^{\beta}) \ge u_\theta^\xf(\mathbf{x}^\theta, j^\theta) \label{eq:op1_c} \\
  & \mathbf{x}^{\theta}, \mathbf{x}^{\beta}\in \Delta_m \label{eq:op1_d}
\end{talign}
\end{subequations}
Namely, we maximize the leader's utility obtained on a type-$\theta$ follower under the condition that the follower is incentivized to report $\beta$.
Here, \eqref{eq:op1_a} and \eqref{eq:op1_b} guarantee that $j^\theta \in \BR_{\theta}(\mathbf{x})$ and $j^{\beta} \in \BR_{\beta}(\mathbf{x})$ as required in the definition of leader policy; and \eqref{eq:op1_c} ensures that reporting $\beta$ is indeed a better choice for a type-$\theta$ follower than reporting $\theta$.
The policy $\sigma_\theta$ is given by the LP with the best optimal solution among all the $n^2$ LPs.

\paragraph{(ii) Utility guarantee.} It remains to show that the best $\sigma_\theta$ over all $\theta \in \Theta$ offers the desired utility.
For each $\theta$, let $\mu_\theta^\xl$ be the utility offered by $\sigma_\theta$.
We have $U^\xl(\sigma_\theta) \ge \pi \cdot \mu_\theta^\xl$ because $\pi \cdot \mu_\theta^\xl$ is only the part of leader utility obtained on type $\theta$.
Now consider an arbitrary feasible policy $\varsigma = \left\{ \langle \mathbf{z}^\theta, k^\theta \rangle \right\}_{\theta\in\Theta}$.
Let $\hat{\theta}(\varsigma)$ denote a type-$\theta$ follower's best reporting strategy in response to policy $\varsigma$.
It holds that $\mu_\theta^\xl \ge u^\xl ( \mathbf{z}^{\hat{\theta}(\varsigma)}, k^{\hat{\theta}(\varsigma)} )$ because $\mathbf{x}^\theta = \mathbf{z}^\theta$ and $\mathbf{x}^\beta = \mathbf{z}^{\hat{\theta}(\varsigma)}$ are feasible under \eqref{eq:op1_a}--\eqref{eq:op1_d} when $j^\theta$ and $j^\beta$ are set to $k^\theta$ and $k^{\hat{\theta}(\varsigma)}$, respectively.
Let $\theta^* = \arg\max_{\theta\in\Theta} \pi_\theta \cdot \mu_\theta^\xl$. It follows that,  for any feasible $\sigma$,
\begin{align*}
U^\xl(\sigma_{\theta^*}) \geq \pi_{\theta^*} \cdot \mu_{\theta^*}^\xl \ge \frac{1}{|\Theta|} \sum_\theta \pi_\theta \cdot \mu_\theta^\xl 
\ge \frac{1}{|\Theta|} \sum_\theta \pi_{\theta} \cdot u^\xl ( \mathbf{x}^{\hat{\theta}(\varsigma)}, j^{\hat{\theta}(\varsigma)} ) = \frac{1}{|\Theta|} \cdot  U^\xl(\varsigma).
\end{align*}
Since the choice of $\varsigma$ is arbitrary, $\sigma_{\theta^*}$ serves as an $\frac{1}{|\Theta|}$-approximation of the optimal policy.
\end{proof}

\subsection{An MILP Formulation and the Tractability with Small $\Theta$}

\newcommand{\var}[1]{\textcolor{black}{#1}}

Given the hardness result, efficient algorithms for computing the optimal policy seems unlikely. We provide an algorithm based on an MILP as a practical approach to tackle the problem. The MILP formulation also forms the basis of a polynomial-time algorithm for {\opt} and {\optic} when the size of the follower type set $\Theta$ is fixed.

First, when IC constraints are  not imposed, the following program (not linear yet) computes the optimal leader policy, where $\mu_\theta^\xl$, $\mathbf{x}^\theta$, $y_{\beta}^\theta$, and $p_j^\theta$ are the variables ($\theta,\beta\in\Theta$, and $j\in[n]$).
\begin{subequations}\label{eq:milp}
\begin{talign}
\text{maximize} \quad &\sum_{\theta  } \pi_{\theta}\ \cdot \var{\mu_\theta^\xl }  \tag{\ref{eq:milp}} \\[0.5mm]
\text{subject to}\quad   %
& \var{\mu_\theta^\xl } \le  \sum_{j} p_j^{\beta} \cdot u^\xl ( \var{\mathbf{x}^{\beta}}, j) + (1 - \var{y_{\beta}^\theta}) \cdot M  && \text{for } \theta, \beta \in \Theta \label{eq:milp_a}\\[0.5mm]
&  p_j^\theta \cdot \left[ u_{\theta}^\xf ( \var{\mathbf{x}^{\theta}}, j) - u_{\theta}^\xf ( \var{\mathbf{x}^{\theta}} , k)  \right] \ge 0  & & \text{for } \theta \in \Theta;\, j,k \in [n] \label{eq:milp_b}\\[0.5mm]
& \sum_j p_j^{\beta} \cdot u_\theta^\xf ( \var{\mathbf{x}^{\beta}}, j ) %
- \sum_j p_j^{\gamma} \cdot u_\theta^\xf ( \var{\mathbf{x}^{\gamma}}, j ) \ge - (1 - \var{ y_{\beta}^\theta }) \cdot M & & \text{for } \theta, \beta, \gamma  \in \Theta \label{eq:milp_c} \\[0.5mm]
& \sum_i \var{x_i^\theta} = \sum_{ j } \var{p_{j}^{\theta}} =  \sum_{\beta} \var{y_{\beta}^{\theta}} = 1 & & \text{for } \theta \in \Theta  \\[0.5mm]
& \var{\mathbf{x}^{\theta}} \in [0,1]^m,\ \ \var{y_{\beta}^\theta} \in \{0,1\},\ \ \var{p_j^\theta} \in \{0,1\} & & \text{for } \theta, \beta \in \Theta ;\, j \in [n]
\end{talign}
\end{subequations}
Here, the leader policy is represented by variables $\mathbf{x}^{\theta}$ and $p_{j}^{\theta}$, where $\mathbf{x}^{\theta}$ is the leader mixed strategy prescribed for report $\theta$, and $p_{j}^{\theta} \in \{0,1\}$ captures the induced follower action.
Through \eqref{eq:milp_b}, it is guaranteed that $p_{j}^{\theta} = 1$ only if $j$ is a best response of type $\theta$ to $\mathbf{x}^{\theta}$, and $p_{j}^{\theta} = 0$, otherwise.
Similarly, $y_{\beta}^{\theta}$ captures the optimal reporting strategy of type $\theta$, i.e., $y_{\beta}^{\theta}=1$ only if reporting $\beta$ is optimal for $\theta$, and this is guaranteed by \eqref{eq:milp_c}, in which $M$ is a sufficiently large constant.
Finally, each $\mu_\theta^\xl$ captures the utility the leader obtains from true type $\theta$ by \eqref{eq:milp_a}.

This program is not yet an MILP because of the quadratic terms $p_j^{\theta} \cdot \mathbf{x}^\theta$.
To linearize the program, we replace these terms with a set of new variables $\tilde{x}_{ji}^\theta$, subject to $0 \le \tilde{x}_{ji}^\theta \le x_i^\theta$ (for all $\theta\in\Theta$, $j\in[n]$, and $i\in[m]$) and $\sum_i \tilde{x}_{ji}^\theta = p_j^\theta$ (for all $\theta\in\Theta$ and $i\in[m]$). This guarantees $\tilde{x}_{ji}^\theta = p_j^{\theta} \cdot {x}_i^\theta$ given that $p_j^{\theta}$ is restricted to be in $\{0,1\}$.

\smallskip

In the situation with IC constraints, we force $y_{\theta}^{\theta} = 1$ for all $\theta \in \Theta$ in the above MILP.
The ``windfall'' of the MILP formulation is Theorem~\ref{thm:opt_fixTheta}, the tractability of {\opt} and {\optic} with a small $\Theta$.

\begin{theorem}
\label{thm:opt_fixTheta}
{For a fixed number of follower types, both {\opt} and {\optic} can be solved in polynomial time.}
\end{theorem}

\begin{proof}
Following Program~\eqref{eq:milp}, when $|\Theta|$ is a constant, the size of the joint space of feasible $\mathbf{y}$ and $\mathbf{p}$ is bounded by a polynomial in $n$.
We can enumerate every combination in this space and solve Program~\eqref{eq:milp} with $\mathbf{y}$ and $\mathbf{p}$ fixed to values in the combination; the program becomes a linear program as the remaining variables are continuous.
\end{proof}

\section{Generalization to Mixed Policy}
\label{sec:mixed_policy}

The leader policy we have considered so far consists of a single \emph{outcome} $\langle \mathbf{x}^{\theta}, j^\theta \rangle$ for each $\theta$. In this section, we extend the definition of leader policy, and allow a policy to prescribed a \emph{distribution} of outcomes for each follower type. We call an outcome distribution a \emph{mixture} and refer to a leader policy consisting of mixtures a \emph{mixed policy}, as opposed to \emph{pure} policies in the previous sections.

Now the game proceeds as follows. The leader first commits to a mixed policy. The follower observes the mixed policy and reports a type $\theta$. The leader then samples an outcome $\langle \mathbf{x}, j \rangle$ from the mixture prescribed for $\theta$, plays $\mathbf{x}$, and induces the follower to respond with action $j$.

Obviously, the leader utility of an optimal mixed policy is at least that of a pure one as all pure policies are also mixed policies by definition. An immediate question is why further randomization over the outcomes will benefit the leader since the leader is already playing a mixed strategy in every outcome, and randomizing mixed strategies normally will not bring any extra power for a player. It turns out that this extra randomization will be beneficial for the leader when there is follower deception. Intuitively, such randomization can help us tame the incentive constraints for the follower's type reporting, while the randomization in mixed strategies is responsible for inducing desired follower action responses.  The following  example provides a more concrete illustration.

\renewcommand{\fone}{\mathsf{1^x}}
\renewcommand{\ftwo}{\mathsf{2^x}}
\renewcommand{\fthr}{\mathsf{3^x}}

\paragraph{Example: Strict Improvement of Mixed Policies}
We consider again a security game example with two targets: areas $1$ and $2$.
The defender can patrol one of these areas and the poacher chooses one to attack.
The poacher has three types $\theta_*$, $\theta_A$ and $\theta_B$ which appear with the equal probability $1/3$. The payoffs matrices are given in the tables below.

\begin{figure}[h!]
\renewcommand{\arraystretch}{1.3}
\newcolumntype{R}{>{\rule[-0.5em]{0pt}{1.5em}$}r<{$}}
\newcolumntype{L}{>{\rule[-0.5em]{0pt}{1.5em}\raggedleft\arraybackslash$}p{5mm}<{$}}
\newcolumntype{C}{>{\rule[-0.5em]{0pt}{1.5em}\centering\arraybackslash$}p{5mm}<{$}}
\centering
\hspace{-2mm}
\begin{tabular}{R|L|L|}
\multicolumn{1}{R}{}&	\multicolumn{1}{C}{\fone}	&  \multicolumn{1}{C}{\ftwo}  \\
\cline{2-3}
 \lone	&	1	& -1	 \\\cline{2-3}
 \ltwo	&	-1	& 1	 \\\cline{2-3}
\multicolumn{1}{R}{}& \multicolumn{2}{C}{\rule[-0.2em]{0pt}{2em} \emph{defender}} \\
\end{tabular}
\hspace{8mm}
\begin{tabular}{R|L|L|}
\multicolumn{1}{R}{}&	\multicolumn{1}{C}{\fone}	&  \multicolumn{1}{C}{\ftwo}  \\
\cline{2-3}
 \lone	&	-1	& 1  \\\cline{2-3}
 \ltwo	&	1	& -1	 \\\cline{2-3}
\multicolumn{1}{R}{}& \multicolumn{2}{C}{\rule[-0.2em]{0pt}{2em} \hspace{-4mm}\emph{poacher type $\theta_*$}} \\
\end{tabular}
\hspace{8mm}
\begin{tabular}{R|L|L|}
\multicolumn{1}{R}{}&	\multicolumn{1}{C}{\fone}	&  \multicolumn{1}{C}{\ftwo}  \\\cline{2-3}
 \lone	&	0	& 0	 \\\cline{2-3}
 \ltwo	&	1	& -2	 \\\cline{2-3}
\multicolumn{1}{R}{}& \multicolumn{2}{C}{\rule[-0.2em]{0pt}{2em} \hspace{-4mm}\emph{poacher type $\theta_A$}} \\
\end{tabular}
\hspace{8mm}
\begin{tabular}{R|L|L|}
\multicolumn{1}{R}{}&	\multicolumn{1}{C}{\fone}	&  \multicolumn{1}{C}{\ftwo}  \\\cline{2-3}
 \lone	&	-2	& 1	 \\\cline{2-3}
 \ltwo	&	0	& 0	 \\\cline{2-3}
\multicolumn{1}{R}{}& \multicolumn{2}{C}{\rule[-0.2em]{0pt}{2em} \hspace{-4mm}\emph{poacher type $\theta_B$}} \\
\end{tabular}
\end{figure}

Note that the two areas are of equal value to both the defender and poacher type $\theta_*$.
Thus, $\theta_*$ will always attack the less frequently patrolled area, so for an outcome $\langle \mathbf{x}, j\rangle$ prescribed for report $\theta_*$, it always holds that $x_j \le 1/2$. Suppose without loss of generality that $j=\fone$. Then type $\theta_A$ is able to obtain at least $1/2$ by reporting $\theta_*$, in which case the actual utility the defender obtains on poacher type $\theta_A$, using a pure policy, cannot be more than $0$, because for a type-$A$ poacher to obtain utility at least $1/2$, the bottom left entries of the payoff matrices need to be chosen with probability at least $1/2$.
The same argument applies to type $\theta_B$ if we suppose $j=\ftwo$.
As a result, any pure policy can obtain utility more than $0$ on at most one of $\theta_A$ and $\theta_B$.
Further, observe that the true payoff parameters of type $\theta_*$ make the game zero-sum. Thus, the actual utility the defender obtains on poacher type $\theta_*$ is at most $0$ because a type-$\theta_*$ poacher is guaranteed utility $0$ by reporting truthfully, and when playing against a type-$\theta_*$ poacher, the defender always gets the negative of the poacher's actual utility.
As a result, a pure policy obtains positive utility on at most one poacher type; her overall utility is at most $1/3$.

Now we consider a {mixed policy} as follows:
\begin{itemize}
  \item
If type $\theta_*$ is reported, choose outcomes $\langle (1/2, 1/2), \fone \rangle$ and $\langle (1/2,1/2), \ftwo \rangle$, each with probability $1/2$;

  \item
If type $\theta_A$ is reported, choose outcome $\langle (1,0), \fone \rangle$ with probability $1$; and if $\theta_B$, choose $\langle (0,1), \ftwo \rangle$ with probability $1$.
\end{itemize}
As such, types $\theta_A$ and $\theta_B$ can only obtain $-1/2$ in expectation if they report $\theta_*$. This incentivizes both of them to report truthfully, yielding expected defender utility $2/3$ --- an improvement of at least $1/3$ compared to the optimal pure policy.

\newcommand{\mixtr}{\phi}

\subsection{Mixtures with Support  Size $n$ Suffice}

To compute the optimal mixed policy, one challenge is that a mixture may be supported on an infinite set since there are infinitely many mixed strategies and hence, as many outcomes.
Fortunately, using a revelation-principle-type argument, we prove that it suffices to consider mixtures supported on at most $n$ outcomes; each outcome induces the follower to choose a distinct action in $[n]$. This is shown in Proposition~\ref{prp:supportsize} below.

\begin{proposition}
\label{prp:supportsize}
For any feasible mixture $\mixtr$ prescribed for a reported type $\theta \in \Theta$, there is a feasible mixture with support size $\tilde{n} \le n$ that yields the same utility as does $\mixtr$, for both the leader and the follower.
\end{proposition}

\begin{proof}
We will show that for any given mixture $\mixtr$ prescribed for report $\theta$ we can merge outcomes that induce the same follower response into one outcome, and this will not change both players' expected utility.

Formally, let $\mathcal{O}$ be the support set of $\mixtr$; let $p_{\langle \mathbf{x}, j \rangle} > 0$ be the probability assigned to outcome $\langle \mathbf{x}, j \rangle \in \mathcal{O}$.
For each $j\in[n]$, let $\mathcal{O}_j = \{ \langle \mathbf{x}, j' \rangle \in \mathcal{O} : j' = j \}$ be the set of outcomes that induce follower action $j$.
Let  $\tilde{p}_j = \sum_{\langle \mathbf{x}, j \rangle \in \mathcal{O}_j } p_{\langle \mathbf{x}, j \rangle } $ and
$\tilde{\mathbf{x}}_j = \frac{1}{\tilde{p}_j} \sum_{\langle \mathbf{x}, j \rangle \in \mathcal{O}_j } p_{\langle \mathbf{x}, j \rangle } \cdot \mathbf{x}$.
Let $\tilde{\mixtr}$ be a mixture supported on $\{ \langle \tilde{\mathbf{x}}_1, 1  \rangle, \dots, \langle \tilde{\mathbf{x}}_n, n  \rangle \}$, with probability $\tilde{p}_j$ for each $\langle \tilde{\mathbf{x}}_j, j  \rangle$.
This new mixture will not change the leader's utility when $\theta$ is reported because of the following:
\begin{eqnarray*}
\sum_{\langle \mathbf{x}, j \rangle \in \mathcal{O} } p_{\langle \mathbf{x}, j \rangle } \cdot u^\xl ( \mathbf{x} , j )
&=& \sum_{j \in [n]} \sum_{\langle \mathbf{x}, j \rangle \in \mathcal{O}_j } p_{\langle \mathbf{x}, j \rangle } \cdot u^\xl ( \mathbf{x} , j ) \\
%
%
&=& \sum_{j \in [n]} \tilde{p}_j  \cdot \ u^\xl \Big(\ \frac{1}{\tilde{p}_j} \sum_{\langle \mathbf{x}, j \rangle \in \mathcal{O}_j } p_{\langle \mathbf{x}, j \rangle } \cdot \mathbf{x}\ ,\ \ j \ \Big) \\
&=& \sum_{j \in [n]} \tilde{p}_j  \cdot u^\xl ( \tilde{\mathbf{x}}_j ,\ j ).
\end{eqnarray*}
By changing labels of the utility function, we can obtain $\sum_{\langle \mathbf{x}, j \rangle \in \mathcal{O} } p_{\langle \mathbf{x}, j \rangle } \cdot u_\beta^\xf ( \mathbf{x} , j ) = \sum_{j \in [n]} \tilde{p}_j  \cdot u_\beta^\xf ( \tilde{\mathbf{x}}_j ,\ j )$ for any follower type $\beta \in \Theta$, so $\tilde{\mixtr}$ will not change the follower's utility, either, irrespective of his true type.

It remains to show that $\tilde{\mixtr}$ is indeed a valid mixture.  Clearly, $\tilde{p}_j \ge 0$ and $\sum_{j} \tilde{p}_j = \sum_{j} \sum_{\langle \mathbf{x}, j \rangle \in \mathcal{O}_j } p_{\langle \mathbf{x}, j \rangle } = \sum_{\langle \mathbf{x}, j \rangle \in \mathcal{O} } p_{\langle \mathbf{x}, j \rangle } = 1$, so $\tilde{p}_j$ is a valid probability distribution.
In addition, we need each  outcome $\langle \tilde{\mathbf{x}}_j, j  \rangle$ to be feasible, i.e., $j \in \BR_{\theta}(\tilde{\mathbf{x}}_j)$, so that $j$ can indeed be induced by $\tilde{\mathbf{x}}_j$. This is true because the set $\mathcal{X}_j = \{\mathbf{x} \in \Delta_m: j \in \BR_{\theta}(\mathbf{x}) \} = \{\mathbf{x} \in \Delta_m: u_{\theta}^\xf (\mathbf{x}, j) \ge u_{\theta}^\xf (\mathbf{x}, j') \text{ for all } j' \in [n] \} $ is a convex set, and by definition, all $\mathbf{x}$ such that $\langle \mathbf{x}, j \rangle \in \mathcal{O}_j$ are in $\mathcal{X}_j$. Thus,  $\tilde{\mathbf{x}}$, as a convex combination of $\mathbf{x} \in \mathcal{X}_j$, also lies in $\mathcal{X}_j$, so $j \in \BR_{\theta}(\tilde{\mathbf{x}}_j)$. Therefore, $\langle \tilde{\mathbf{x}}_j, j \rangle$ is a feasible outcome and $\tilde{\mixtr}$ a valid mixture. This completes the proof.
\end{proof}

Following this result, we can represent the mixture for each type $\theta$ as a vector $(p_1^\theta, \dots, p_n^\theta) \in \Delta_n$, along with a support set $\{\mathbf{x}_1^\theta,\dots,\mathbf{x}_n^\theta\}$ of $n$ mixed strategies; the mixture samples each outcome $\langle \mathbf{x}_j^\theta, j \rangle$ with probability $p_j^\theta$.
Note that this representation allows outcomes involved to be invalid because it is possible that some follower action $j$ cannot be induced by any leader strategy, i.e., $j \notin \BR_{\theta}(\mathbf{x})$ for all $\mathbf{x}\in \Delta_m$.
Thus, a mixture is valid only if $j \in \BR_{\theta}(\mathbf{x}_j^\theta)$ for all $j$ such that $p_j^\theta > 0$.

\subsection{Computing the Optimal Mixed Policy}

Proposition \ref{prp:supportsize} implies that there exists an optimal mixed policy of polynomial size. Nevertheless, out next theorem, Theorem~\ref{thm:apx_optBay_mix}, shows that the problem of computing the optimal mixed policy, referred to as {\optx}, remains inapproximable in the situation \emph{without} IC constraints.
Interestingly, after we impose IC constraints, the problem {\optxic} becomes \emph{tractable}. We present these results next.

\begin{theorem}
  \label{thm:apx_optBay_mix}
 For any constant $\epsilon > 0$, {\optx} does not admit any polynomial-time $\frac{1} {{(|\Theta|-1)}^{1-\epsilon}}$-approximate algorithm unless P = NP, even when the follower has only three actions.
\end{theorem}
\begin{proof}
We show a reduction from the same \textsc{Max-Independent-Set} problem as in the proof of Theorem~\ref{thm:apx_optBay}, but we will need to adapt the correctness proof to {mixed} policies here.

For one direction of the reduction, if there is an independent set $V'$ of size $k$ in $G$, the same policy we constructed in the proof of Theorem~\ref{thm:apx_optBay} achieves overall leader utility $\frac{k}{|\Theta|-1}$ for the same argument.

For the other direction, observe that it holds still with mixed policies that the leader gets $1$ when $\theta_*$ is reported, and $0$, otherwise; hence, the leader's overall utility is proportional to the number of types $\theta_v$, $v\in V$, who report $\theta_*$.
We show that, for any mixed policy $\sigma$,  $\{ v \in V: \hat{\theta}_v (\sigma) = \theta_*\}$ is an independent set to conclude the proof.
Suppose for a contradiction that $\hat{\theta}_v (\sigma) = \theta_*$ and $\hat{\theta}_{v'} (\sigma) = \theta_*$ but $(v,v')\in E$.
We make the following observations:
  \begin{itemize}
    \item[(i)] $p_{\fone}^{\theta_{v}} = 0$ because $\fone$ is strictly dominated by $\ftwo$ for a follower of type $\theta_{v}$ and hence $\fone \notin \BR_{\theta_{v}}(\mathbf{x})$ for any leader strategy $\mathbf{x}$. It follows that we need $p_{\ftwo}^{\theta_{v}} > 0$ or $p_{\fthr}^{\theta_{v}} > 0$, as $(p_{\fone}^{\theta_{v}}, p_{\ftwo}^{\theta_{v}}, p_{\fthr}^{\theta_{v}}) \in \Delta_3$.

    \item[(ii)] Since a follower of type $\theta_{v'}$ gets utility $0.5$ by reporting $\theta_*$, we need $p_{\ftwo}^{\theta_{v}} \cdot x_{\ftwo,a_{v}}^{\theta_{v}} = 0$ and $p_{\fthr}^{\theta_{v}} \cdot x_{\fthr,a_{v}}^{\theta_{v}} = 0$ as otherwise the action profile  $\langle a_{v}, \ftwo \rangle$ would be chosen with some probability and $\theta_{v'}$ would be better-off reporting $\theta_{v}$ to get more than $0.5$.
  \end{itemize}
  Now, consider the two possible cases implied by (i).
  \begin{itemize}
    \item If $p_{\fthr}^{\theta_{v}} > 0$, we then have $x_{\fthr,a_{v}}^{\theta_{v}} = 0$ by (ii), which means some action profile $\langle i, \fthr \rangle$, $i\neq a_{v}$, would be chosen with some probability, in which case a follower of type $\theta_{v}$ would be better-off reporting truthfully instead of reporting $\theta_*$ --- a contradiction.

    \item If $p_{\ftwo}^{\theta_{v}} > 0$, we have $x_{\ftwo,a_{v}}^{\theta_{v}} = 0$ by (ii). Now in order for a type-$\theta_{v}$ follower to still find reporting $\theta_*$ optimal, we need $x_{\ftwo, i}^{\theta_{v}} = 0$ for all $i\in \{a_0, b_{v}\} \cup \mathcal{N}({v})$ (correspondingly, the first, third and fourth rows of the payoff matrix) as otherwise truthful report would gain type $\theta_{v}$ a higher utility; consequently, only the fifth row of the payoff matrix will be chosen by $\mathbf{x}_{\ftwo}^{\theta_{v}}$, in which case, however, $\BR_{\theta_{v}} (\mathbf{x}_{\ftwo}^{\theta_{v}}) = \{\fthr\}$, so $\langle \mathbf{x}_{\ftwo}^{\theta_{v}}, \ftwo \rangle$ cannot be a valid outcome --- a contradiction to the assumption that $p_{\ftwo}^{\theta_{v}} > 0$.
  \end{itemize}
  This completes the proof.
\end{proof}

The bound in the above inapproximability result is tight following a similar argument for {\opt}.
Moreover, an MILP formulation for {\optx} can be obtained by modifying Program~\eqref{eq:milp}, and this yields a polynomial-time algorithm for {\optx} when the number of follower types is fixed.
If we further impose IC constraints, the constraints remove the only integer variables $y_{\beta}^\theta$ in the MILP (variables $p_j^\theta$ are relaxed now) and result in a linear program; this establishes the tractability of {\optxic}.
We sketch the results in Theorems~\ref{thm:optx_apx}--\ref{thm:optBay_mix_ic} below.

\begin{theorem}
\label{thm:optx_apx}
  {\optx} admits a polynomial-time  $\frac{1}{|\Theta|}$-approximation algorithm.
\end{theorem}

\begin{proof}[Proof Sketch]
  The proof follows from a similar argument for Theorem~\ref{thm:1_Theta_apx}.
\end{proof}

\begin{theorem}
\label{thm:optx_fixTheta}
  {For a fixed number of follower types, {\optx} can be solved in polynomial time.}
\end{theorem}

\begin{proof}[Proof Sketch]
We modify Program~\eqref{eq:milp} for a formulation for {\optx}: replace every $\mathbf{x}^\theta$ by $\mathbf{x}_j^\theta$ wherever they are associated with $p_j^\theta$; and relax every $p_j^\theta$ to be in $[0,1]$.
The modified program can be linearized into an MILP in the same way we linearize Program~\eqref{eq:milp}.
By the same argument for Theorem~\ref{thm:optx_fixTheta}, the MILP formulation yields a polynomial-time algorithm for {\optx} when the number of follower types is fixed.
\end{proof}

\begin{theorem}
\label{thm:optBay_mix_ic}
{\optxic}  can be computed in polynomial time.
\end{theorem}

\begin{proof}[Proof Sketch]
The problem can be formulated as a linear program as follows,
where $\tilde{\mathbf{x}}_j^{\theta}$ and $p_j^\theta$ are the variables ($j\in[n], \theta \in \Theta$).
\begin{talign*}
\text{maximize} \quad & \sum_{\theta  } \pi_{\theta}\ \sum_j u^\xl ( \var{\tilde{\mathbf{x}}_j^{\theta}}, j)  \\[0.5mm]
\text{subject to} \quad
&  \sum_j  u_\theta^\xf ( \var{\tilde{\mathbf{x}}_j^{\theta}}, j ) %
- \sum_j  u_\theta^\xf ( \var{\tilde{\mathbf{x}}_j^{\beta}}, j ) \ge 0  & &\text{for } \theta, \beta \in \Theta \\
&  u_{\theta}^\xf ( \var{\tilde{\mathbf{x}}_j^{\theta}}, j) - u_{\theta}^\xf ( \var{\tilde{\mathbf{x}}_j^{\theta}} , k)  \ge 0  & &\text{for }  \theta \in \Theta; \, j,k \in [n] \\
& \sum_{ i } \var{\tilde{x}_{ji}^{\theta}} =  \var{p_j^\theta} & &\text{for }  \theta \in \Theta;\, j\in [n]  \\
& \var{\tilde{\mathbf{x}}_j^\theta} \in \mathbb{R}_{\ge 0}^m,  \ \var{(p_1^\theta, \dots, p_n^\theta)} \in \Delta_n & &\text{for } \theta, \beta \in \Theta ;\, j \in [n]
\end{talign*}
The strategies $\mathbf{x}_j^\theta$ in the leader policy can be extracted from the solution as $\mathbf{x}_j^\theta = \frac{1}{p_j^\theta} \cdot \tilde{\mathbf{x}}_j^{\theta}$.
\end{proof}

\section{Empirical Evaluation}
\label{sec:experiment}

We evaluate our framework with games generated randomly using the well-known covariance game model~\cite{nudelman2004run}.
We generate the players' payoffs uniformly at random from the range $[0,1]$, and then adjust the follower's payoffs by blending them with the negative value of the leader's payoffs;
the degree of blending is controlled by a parameter $\alpha \in [0,1]$, i.e., $u_\theta^\xf \leftarrow (1-\alpha) \cdot u_\theta^\xf - \alpha \cdot u^\xl$.
As such, the game is zero-sum when $\alpha=1$; and, when $\alpha = 0$, payoffs of the leader and the follower are completely uncorrelated.
We also generate the probabilities $\pi_\theta$'s uniformly at random.


\newlength\figurewidth
\setlength\figurewidth{.32\linewidth}
\newlength\figureheight
\setlength\figureheight{.29\linewidth}


\newenvironment{rtplot}[1]
{
	\begin{tikzpicture}
	\begin{axis}[
	width=\figurewidth,
	height=\figureheight,
	xlabel={\empty},
	ylabel={\empty},
	xmin=2, xmax=20,
	ymin=0.75, ymax=1.15,
    xtick={2,8,14,20},
    xticklabels={$2$,$8$,$14$,$20$},
	ytick={0.7,0.8,0.9,1.0,1.1,1.2,1.3,1.4},
	yticklabels={$.7$,$.8$,$.9$,$1.0$,$1.1$,$1.2$, $1.3$, $1.4$},
    tick label style = {font={\small}},
	major tick length=0.6mm,
	every tick/.style={
		black,
	},
	legend pos=north east,
	ymajorgrids=false,
	grid style=dotted,
	style={font=\small},
	legend style={font=\small, /tikz/every even column/.append style={column sep=3mm}, legend columns=7, inner sep=0.5mm, anchor=north east, },
	legend cell align=left,
	x label style={font={\small},at={(axis description cs:0.5,0.0)}},
	y label style={font={\small},at={(axis description cs:0.1,0.48)}},
	cycle list={%
		{blue}, 
		{blue, mark=square*, mark options={solid, thin}, mark size=0.8pt}, 
		{blue, densely dotted, semithick}, 
		{blue, densely dotted, semithick, mark=x, mark options={solid, thin}, mark size=1.5pt}, 
		{black, densely dashdotted, mark=triangle*, mark options={solid, thin}, mark size=1.5pt, semithick}, 
		{red, densely dashed, semithick}, 
		{black, densely dashdotted, semithick}, 
	},
	title style={font={\small}, at={(axis description cs:0.5,-0.7)}},
	#1,
	]
}{
\end{axis}
\end{tikzpicture}
}

\newenvironment{rtplotxalpha}[1]
{
	\begin{rtplot}
    {			
        xmin=0, xmax=1,
		xtick={0, 0.2, 0.4, 0.6, 0.8, 1.0},
		xticklabels={$0$, $.2$, $.4$, $.6$, $.8$, $1$},
        #1,
    }
}
{
\end{rtplot}
}

\begin{figure*}[t]
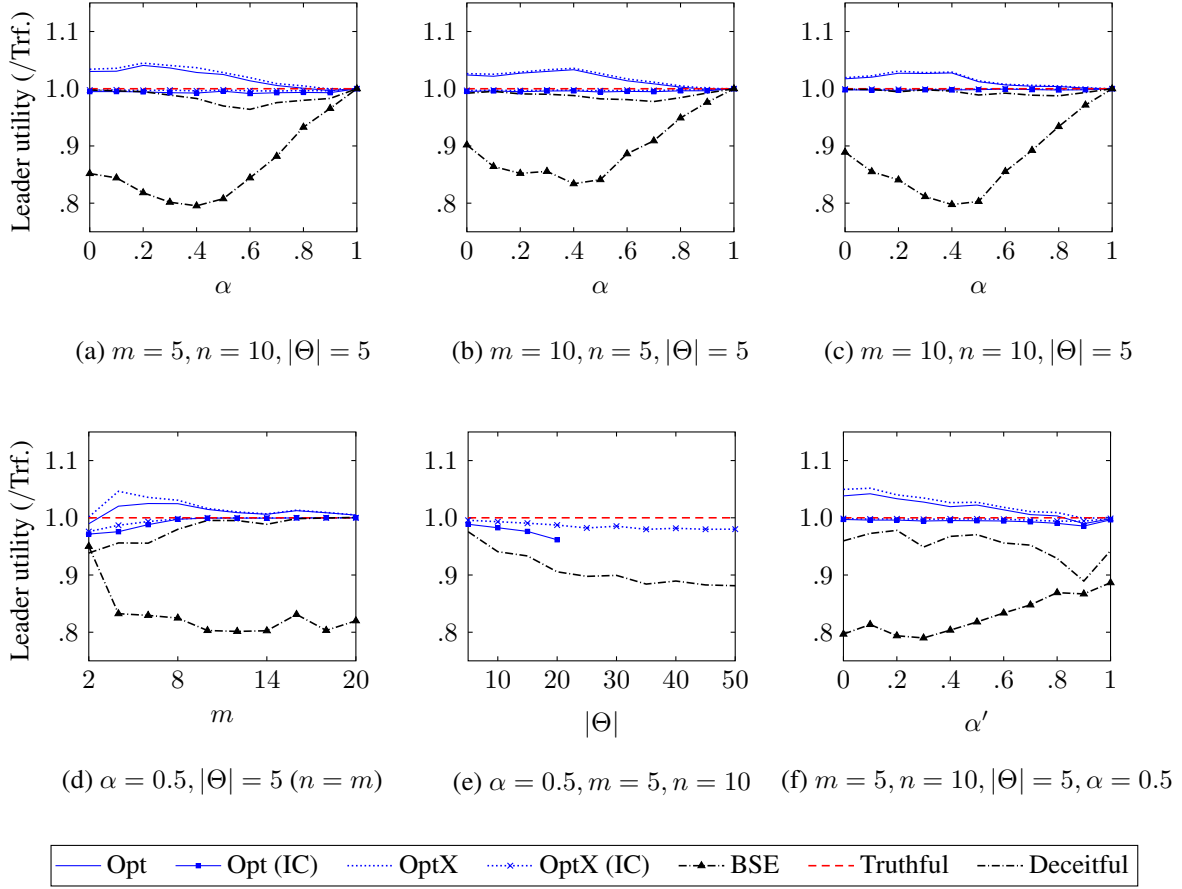

\newlength\seplength
\setlength\seplength{-4.5mm}
\newcommand{\axessetting}{
    xmin=5,xmax=50,
    xtick={10,20,30,40,50},
    xticklabels={$10$,$20$,$30$,$40$,$50$},
    }
\newcommand{\legendsetting}{
    cycle list={%
		{blue, mark=square*, mark options={solid, thin}, mark size=0.8pt}, 
		{blue, densely dotted, semithick, mark=x, mark options={solid, thin}, mark size=1.5pt}, 
		{red, densely dashed, semithick}, 
		{black, densely dashdotted, semithick}, 
        },
    }
\center
\hspace{-8mm}	
\begin{rtplotxalpha}
    {
        title={(a) $m=5,n=10, |\Theta|=5$},
        xlabel={$\alpha$},
		ylabel={Leader utility ($/$Trf.)},
		legend entries={ {Opt}, {Opt (IC)}, {OptX}, {OptX (IC)}, {BSE}, {Truthful}, {Deceitful} },%
		legend to name=leg:all			
    }
\addplot coordinates{	(0,1.0301580983029)	(0.1,1.03059739321809)	(0.2,1.04076342237389)	(0.3,1.03629487865362)	(0.4,1.02842619876272)	(0.5,1.02486198120083)	(0.6,1.01382669177492)	(0.7,1.00566585192104)	(0.8,1.0008471297292)	(0.9,0.996858929092428)	(1,1)	};	
\addplot coordinates{	(0,0.994957346466641)	(0.1,0.995014306253699)	(0.2,0.994518272077888)	(0.3,0.993033169836099)	(0.4,0.992517918293578)	(0.5,0.995641919746188)	(0.6,0.991597403373848)	(0.7,0.993037289833216)	(0.8,0.993535399690255)	(0.9,0.993074456834213)	(1,1)	};	
\addplot coordinates{	(0,1.03410256440955)	(0.1,1.03571991053646)	(0.2,1.04478617786455)	(0.3,1.04064301590511)	(0.4,1.0367705018615)	(0.5,1.02822352743047)	(0.6,1.01930152907741)	(0.7,1.00893310110929)	(0.8,1.00494362419962)	(0.9,0.99939847817168)	(1,1)	};	
\addplot coordinates{	(0,0.997596839413399)	(0.1,0.997929482444308)	(0.2,0.997183465573958)	(0.3,0.997476717671493)	(0.4,0.996258607333014)	(0.5,0.997867761896468)	(0.6,0.996342135155072)	(0.7,0.996100504834148)	(0.8,0.996791686000955)	(0.9,0.996443018410604)	(1,1)	};	
\addplot coordinates{	(0,0.851773807470239)	(0.1,0.844194373306091)	(0.2,0.818296484392859)	(0.3,0.801654904669501)	(0.4,0.795385716667136)	(0.5,0.80778092187877)	(0.6,0.844215988011361)	(0.7,0.881934285317567)	(0.8,0.932769728668685)	(0.9,0.965538690999915)	(1,1)	};	
\addplot coordinates{	(0,1)	(0.1,1)	(0.2,1)	(0.3,1)	(0.4,1)	(0.5,1)	(0.6,1)	(0.7,1)	(0.8,1)	(0.9,1)	(1,1)	};	
\addplot coordinates{	(0,0.997791059956657)	(0.1,0.996202550205554)	(0.2,0.994871303478072)	(0.3,0.988838260286155)	(0.4,0.98293433236245)	(0.5,0.969880863731002)	(0.6,0.963969935630659)	(0.7,0.975714508089903)	(0.8,0.979769517383461)	(0.9,0.983020782397605)	(1,1)	};	
	\end{rtplotxalpha}
	%
	\hspace{\seplength}
	\begin{rtplotxalpha}
    {
        title={(b) $m=10,n=5, |\Theta|=5$},
        xlabel={$\alpha$},
    }
\addplot coordinates{	(0,1.02401715690314)	(0.1,1.02176793376673)	(0.2,1.0271598913551)	(0.3,1.03024465188528)	(0.4,1.0332233220904)	(0.5,1.02317458324897)	(0.6,1.01382519016114)	(0.7,1.0085531727498)	(0.8,1.00284082838092)	(0.9,0.99890879692751)	(1,1)	};	
\addplot coordinates{	(0,0.995309820083101)	(0.1,0.996175274832256)	(0.2,0.994779966330171)	(0.3,0.996045935241794)	(0.4,0.996303300407356)	(0.5,0.994297511577393)	(0.6,0.99535746975907)	(0.7,0.994860328024228)	(0.8,0.996371432589176)	(0.9,0.996450253530061)	(1,1)	};	
\addplot coordinates{	(0,1.02635973530637)	(0.1,1.02533154775442)	(0.2,1.02908175678032)	(0.3,1.03316783663354)	(0.4,1.03565992586896)	(0.5,1.02659705481383)	(0.6,1.01710184881171)	(0.7,1.0114955668312)	(0.8,1.00503563036878)	(0.9,1.00052974676985)	(1,1)	};	
\addplot coordinates{	(0,0.996980663781774)	(0.1,0.99750217525349)	(0.2,0.99671580352145)	(0.3,0.997635903944781)	(0.4,0.997679385716808)	(0.5,0.99641212231947)	(0.6,0.997484730324478)	(0.7,0.99706405143713)	(0.8,0.997916272528612)	(0.9,0.998251446360951)	(1,1)	};	
\addplot coordinates{	(0,0.901808331156212)	(0.1,0.864110353621326)	(0.2,0.852088302171901)	(0.3,0.855549292049135)	(0.4,0.833984007152753)	(0.5,0.841108285161714)	(0.6,0.886325002717359)	(0.7,0.909003254663037)	(0.8,0.948961444589501)	(0.9,0.976360279617117)	(1,1)	};	
\addplot coordinates{	(0,1)	(0.1,1)	(0.2,1)	(0.3,1)	(0.4,1)	(0.5,1)	(0.6,1)	(0.7,1)	(0.8,1)	(0.9,1)	(1,1)	};	
\addplot coordinates{	(0,0.992196775373215)	(0.1,0.99535792974939)	(0.2,0.991001656460762)	(0.3,0.990384479596007)	(0.4,0.988054996394488)	(0.5,0.98232540156731)	(0.6,0.981017050661452)	(0.7,0.977632196380917)	(0.8,0.984402195354149)	(0.9,0.992540679627349)	(1,1)	};	
	\end{rtplotxalpha}
	\hspace{\seplength}
    \begin{rtplotxalpha}
    {
        title={(c) $m=10,n=10, |\Theta|=5$},
        xlabel={$\alpha$},
    }
\addplot coordinates{	(0,1.01733389100879)	(0.1,1.02020201775343)	(0.2,1.02707077915959)	(0.3,1.02669107334518)	(0.4,1.02755539267927)	(0.5,1.01223001564498)	(0.6,1.00668715392461)	(0.7,1.00426821632581)	(0.8,1.00388230344737)	(0.9,1.00105033238135)	(1,1)	};	
\addplot coordinates{	(0,0.998058821270287)	(0.1,0.997763890868862)	(0.2,0.997162502682546)	(0.3,0.998208467757201)	(0.4,0.997881179583443)	(0.5,0.998390145043544)	(0.6,0.999398872807847)	(0.7,0.998485378756751)	(0.8,0.997812392227028)	(0.9,0.998197370812812)	(1,1)	};	
\addplot coordinates{	(0,1.01928036116741)	(0.1,1.02259750114515)	(0.2,1.03076374832089)	(0.3,1.02811280065665)	(0.4,1.02905970848266)	(0.5,1.01380124899784)	(0.6,1.00768710615052)	(0.7,1.00601780330776)	(0.8,1.00487407473947)	(0.9,1.00167190439459)	(1,1)	};	
\addplot coordinates{	(0,0.999256030609612)	(0.1,0.999062453344134)	(0.2,0.998875366706851)	(0.3,0.999014783505778)	(0.4,0.998755323844403)	(0.5,0.99912204488211)	(0.6,0.99978165298388)	(0.7,0.999263053370237)	(0.8,0.998765032856186)	(0.9,0.999120472671571)	(1,1)	};	
\addplot coordinates{	(0,0.889098061881654)	(0.1,0.85505809657182)	(0.2,0.840689326130563)	(0.3,0.811516950834365)	(0.4,0.797586222927976)	(0.5,0.802965290866778)	(0.6,0.855312036832639)	(0.7,0.892256025909594)	(0.8,0.933910512994635)	(0.9,0.971503493700271)	(1,1)	};	
\addplot coordinates{	(0,1)	(0.1,1)	(0.2,1)	(0.3,1)	(0.4,1)	(0.5,1)	(0.6,1)	(0.7,1)	(0.8,1)	(0.9,1)	(1,1)	};	
\addplot coordinates{	(0,0.998758399445919)	(0.1,0.998792885529646)	(0.2,0.994694183907098)	(0.3,0.998095750211217)	(0.4,0.995857183117483)	(0.5,0.989031964836399)	(0.6,0.992322488785724)	(0.7,0.988705891718021)	(0.8,0.987608564194809)	(0.9,0.993434530189964)	(1,0.9999994756222)	};	
	\end{rtplotxalpha}
	\\[7mm]
\setlength\seplength{-3mm}
	\begin{rtplot}{title={(d) $\alpha = 0.5, |\Theta|=5$ ($n=m$)},
xlabel={$m$},
ylabel={Leader utility ($/$Trf.)},
}
\addplot coordinates{	(2,0.989562836986694)	(4,1.02029566191407)	(6,1.02492436607072)	(8,1.02467066210822)	(10,1.01466055442446)	(12,1.00900113064848)	(14,1.0061557479772)	(16,1.01247207506431)	(18,1.00880175090707)	(20,1.00474906414475)		};	
\addplot coordinates{	(2,0.971511971095546)	(4,0.975701821966559)	(6,0.988030476231553)	(8,0.997480235496177)	(10,0.999456668748487)	(12,0.999519675925193)	(14,0.999362964672572)	(16,0.99980744010923)	(18,0.999915271203153)	(20,0.999994892552064)		};	
\addplot coordinates{	(2,1.00132125924919)	(4,1.04629217888983)	(6,1.03568170534679)	(8,1.03071860437399)	(10,1.01627488416467)	(12,1.01032270178456)	(14,1.00677541069259)	(16,1.01332165452918)	(18,1.0093747855331)	(20,1.0049565145874)		};	
\addplot coordinates{	(2,0.975892737905239)	(4,0.986876742345809)	(6,0.993413472851655)	(8,0.998144599714986)	(10,0.999762556450393)	(12,0.999748944242483)	(14,0.999672625054297)	(16,0.999937542523805)	(18,0.99993691016373)	(20,0.999997446276032)		};	
\addplot coordinates{	(2,0.950114178956337)	(4,0.832395567969704)	(6,0.829455682287545)	(8,0.824772367245017)	(10,0.802986550180219)	(12,0.801447838752075)	(14,0.802579237627016)	(16,0.83108381023067)	(18,0.803171377371985)	(20,0.819990620817469)		};	
\addplot coordinates{	(2,1)	(4,1)	(6,1)	(8,1)	(10,1)	(12,1)	(14,1)	(16,1)	(18,1)	(20,1)		};	
\addplot coordinates{	(2,0.938463080775791)	(4,0.956062750244657)	(6,0.95564014937444)	(8,0.979894344168513)	(10,0.995386247078477)	(12,0.995141857481956)	(14,0.988637979624019)	(16,0.998472856116894)	(18,0.999627061465284)	(20,1.00069205919532)		};	
	\end{rtplot}
\hspace{\seplength}
	\begin{rtplot}{
            title={(e) $\alpha = 0.5, m=5, n=10$},
            xlabel={$|\Theta|$},
            \axessetting,
            \legendsetting,
    }
\addplot coordinates{	(5,0.988704989843147)	(10,0.982604253763828)	(15,0.976182719780049)	(20,0.961580053616199)								};	
\addplot coordinates{	(5,0.995685442138751)	(10,0.99316635519636)	(15,0.990558541221693)	(20,0.987297291341296)	(25,0.982208799374051)	(30,0.985345807443603)	(35,0.979841192725742)	(40,0.981669933045257)	(45,0.979925408984085)	(50,0.980141896049943)		};	
\addplot coordinates{	(5,1)	(10,1)	(15,1)	(20,1)	(25,1)	(30,1)	(35,1)	(40,1)	(45,1)	(50,1)		};	
\addplot coordinates{	(5,0.975676512992605)	(10,0.940744580974127)	(15,0.933321032691538)	(20,0.905571174616643)	(25,0.897628066301095)	(30,0.899350278381568)	(35,0.88414692159626)	(40,0.889487925204866)	(45,0.882653119798301)	(50,0.881308239995978)		};	
	\end{rtplot}
	\hspace{\seplength}
	\begin{rtplotxalpha}
    {
        title={(f) $m=5,n=10, |\Theta|=5, \alpha=0.5$},
        xlabel={$\alpha'$},
    }
\addplot coordinates{	(0,1.03827493392174)	(0.1,1.04197656352738)	(0.2,1.03333344956636)	(0.3,1.02747435657403)	(0.4,1.01940956640115)	(0.5,1.02225629443415)	(0.6,1.01391664024396)	(0.7,1.00557086295655)	(0.8,1.00328577089885)	(0.9,0.989843942718021)	(1,0.998140895664036)	};	
\addplot coordinates{	(0,0.997092159686379)	(0.1,0.995957390283109)	(0.2,0.996025548983773)	(0.3,0.994359245304947)	(0.4,0.995150989373314)	(0.5,0.994866332348553)	(0.6,0.994436784433784)	(0.7,0.992995453542375)	(0.8,0.990342085624901)	(0.9,0.985528652661793)	(1,0.996438550524545)	};	
\addplot coordinates{	(0,1.04963239903957)	(0.1,1.05165909383188)	(0.2,1.04006655529769)	(0.3,1.03459440745923)	(0.4,1.02647873419686)	(0.5,1.0271088089525)	(0.6,1.01831568355481)	(0.7,1.01081966613211)	(0.8,1.00917856531767)	(0.9,0.999049925325457)	(1,1.00054011779513)	};	
\addplot coordinates{	(0,0.998692397635015)	(0.1,0.997769933858621)	(0.2,0.998187171322725)	(0.3,0.997349466791951)	(0.4,0.997980541805928)	(0.5,0.997979953150501)	(0.6,0.997506730521284)	(0.7,0.995802217235321)	(0.8,0.994191409482768)	(0.9,0.994111207051703)	(1,0.998837143971729)	};	
\addplot coordinates{	(0,0.796684110454751)	(0.1,0.813471349282677)	(0.2,0.793810416644501)	(0.3,0.790011740558244)	(0.4,0.803649750472524)	(0.5,0.817957183182604)	(0.6,0.83358321303992)	(0.7,0.84770108701363)	(0.8,0.869159924105334)	(0.9,0.866948611541678)	(1,0.886550505586744)	};	
\addplot coordinates{	(0,1)	(0.1,1)	(0.2,1)	(0.3,1)	(0.4,1)	(0.5,1)	(0.6,1)	(0.7,1)	(0.8,1)	(0.9,1)	(1,1)	};	
\addplot coordinates{	(0,0.959694528569613)	(0.1,0.972654444841698)	(0.2,0.978018987010655)	(0.3,0.948850366336171)	(0.4,0.967500630179741)	(0.5,0.970483795299745)	(0.6,0.955983145028786)	(0.7,0.952194715458887)	(0.8,0.928974637126616)	(0.9,0.889455777180201)	(1,0.94260496815723)	};	
	\end{rtplotxalpha}
	\\[5mm]
\ref{leg:all}\\[3mm]
	\caption{Leader utility obtained with different approaches.
All y-axes represent leader utility as a \emph{ratio to the utility obtained in the truthful situation}. The missing data points are instances that are  not able to be solved within one hour in our platform because of the scalability issue. All results are the average of 50 runs.
}\label{fig:exp}
\end{figure*}

We compare the leader's utility obtained by different policies/approaches proposed in this paper and previous research, including:
(1) Optimal pure policy (labeled \emph{Opt}), with and without IC;
(2) Optimal mixed policy (\emph{OptX}), with and without IC;
(3) Bayesian Stackelberg Equilibrium (\emph{BSE});
(4) Optimal leader strategy when the follower truthfully reports his type (\emph{Truthful}), and (5) the same strategy but when the follower strategically reports his type (\emph{Deceitful}).

Figure~\ref{fig:exp} depicts our results.
The first three figures, (a)--(c), show the variance of leader utility with $\alpha$ under different numbers of player actions.
For ease of comparison, the results are depicted as ratios to the leader utility in the truthful situation.
The figures exhibit quite similar patterns.
All our proposed polices (lines in blue) improve the leader's utility significantly upon BSE. The utilities obtained are also very close to the truthful utility even when IC constraints are imposed. When IC is not required, the optimal policies even perform better in almost all experiments. Surprisingly, even when the leader na\"ively trusts a deceitful follower, the results are very positive; the utility obtained is usually very close to the truthful utility. This is in contrast to our theoretical examples. An explanation is that in randomly generated games, the negative effect of follower deception is likely to be cancelled out by positive effect brought by benign follower types: as we have observed in Section~\ref{sec:model}, follower deception is not always bad for the leader.

The utility differences are the most significant when $\alpha$ is around $0.5$.
Intuitively, this is because when $\alpha$ approaches $1$ the diminishing uncertainty over follower types reduces the effect of follower deception.
But when $\alpha$ is close to $0$, payoffs of the leader and the follower are less correlated and their objectives less conflicted; hence, when the follower leverages deception to increase his utility, there is a lower chance that this will decrease the leader's utility.

We further investigate the effects of varying numbers of actions and of follower types.
From Figure~\ref{fig:exp}~(d) we see that utility differences (except for BSE) quickly diminish when the number of actions increase. Similar to the situation when $\alpha$ is close to $0$, the level of randomness of the payoffs increases as the action space grows, which dissolves the negative effect of follower deception.
In contrast, as shown in Figure~\ref{fig:exp}~(e), the effect of follower deception appear to increase with the number of follower types as leader utility in the deceitful setting drops.
Utility offered by the optimal mixed IC policy, however, decreases much slower. Hence, the optimal mixed IC policy offers a scalable practical solution when there is a large number of follower types. At a high level, these results suggest that uncertainty over follower's payoff information is a major amplifier of the negative effect of follower deception.

To explore situations in which follower deception is likely to be more harmful, we slightly modify the covariance game model, generating the follower type set with multiple blending parameters.
In Figure~\ref{fig:exp}~(f), the results are obtained on a follower type set in which half of the types are generated with $\alpha=0.5$, and another half generated with $\alpha'$ as on the x-axis. Interestingly, the difference between the approaches suddenly becomes more distinguishable, and there is a clear gap between the optimal and the deceptive situations.
An in-depth characterization of when deception tends to be more harmful to the leader and when it is not would be an interesting question for future work.

\begin{table}[t]
\renewcommand{\arraystretch}{1.2}
\begin{center}
\begin{tabular}{llllll}
\noalign{\hrule height 1pt}
& {$10^{-5}$} & {$10^{-4}$} & {$10^{-3}$} & {$10^{-2}$} & {$10^{-1}$}  \\
\noalign{\hrule height 0.5pt}
Opt                 &1.00   &1.00	&1.00	&0.97	&0.63  \\
Opt (IC)            &1.00	&1.00	&1.00	&0.97	& 0.61$^*$ 	\\
OptX                &1.00	&1.00	&1.00	&0.97	&0.66 	\\
OptX (IC)           &1.00	&1.00	&1.00	&0.98	&0.67  \\
OptX (IC) - large	&1.00	&1.00	&1.00	&0.98	&0.79  \\
\noalign{\hrule height 1pt}
\end{tabular}
\end{center}
\caption{Leader utility yielded by robust solutions as a ratio to the utility under the optimistic tie-breaking assumption, with $\epsilon$ varying from $10^{-5}$ to $10^{-1}$. Results of \emph{OptX (IC) - large} are obtained on large instances with $\alpha=0.5$, $|\Theta|=50$, $n=m=20$; all other results are obtained on instances with $\alpha=0.5$, $|\Theta|=5$, $n=5$ and $m=10$. Each entry is the average of 100 runs. ($*$) In 9 out of 100 runs of this setting there is no feasible solution.}\label{tbl:robust}
\end{table}

\paragraph{Robustness to Tie-breaking}
In our model, we assumed that the follower will break ties in favor of the leader.
The assumption is already widely adopted in the literature of Stackelberg games because such favorable tie-breaking behavior can typically be induced by an infinitesimal deviation of the leader's strategy, though in some rare and extreme cases such inducibility may not hold true (see a recent discussion in \cite{guo18inducible}).
Our next set of simulation results experimentally demonstrate that inducibility is \emph{not} an issue in our model for naturally generated instances. 
In the simulations, we enforce a small utility gap of at least $\epsilon$ for the follower between taking the induced action and any other action (the same with the induced reporting strategy). This removes all the ties, and makes the follower \emph{strictly} prefer the action (by $\epsilon$) the leader intends to induce.
The results are depicted in Table~\ref{tbl:robust}, with $\epsilon$ varying from $10^{-5}$ to $10^{-1}$. The experimental results show that our solution is  robust to tie breaking and always maintains optimality for $\epsilon \leq 10^{-3}$ (recall that the follower's payoffs are in $[-1,1]$). Moreover, desired follower actions are almost always inducible only except for a small number of instances when $\epsilon=0.1$.\footnote{Note that we list $\epsilon = 0.1$ only for illustrative purpose as in practice one would never want to use such a big $\epsilon$ to induce tie breaking. There are extended approaches able to generate feasible robust solutions for every $\epsilon \in [0,1]$ (see~\cite{pita2010robust}), but this is not the focus of this paper.}
These results empirically show that the issue of tie-breaking does not affect the validity of our approach in natural instances.

\section{Conclusion}
\label{sec:conclusion}

In this paper, we point out potential occurrences of (imitative) follower deception in Stackelberg games and the risk when they are ignored.
We then propose a framework to design leader policies against such deception and provide a systematic study of the computational aspects of this framework.
Our results shows that handling follower deception is hard in general, even when the searching of optimal policy is expanded into the less complicated space of mixed policies.
Our experiments indicate that the loss due to ignoring follower deception grows with uncertainty over follower types. However, when there is only a small number of types, the average loss appear to be very small compared to results in the theoretical analysis.

There are a number of potential directions for future work. Perhaps the first ones to investigate are several extended settings, including the setting with no prior knowledge about the distribution of follower types, and the setting where follower types fall into a continuous space (e.g., when each payoff parameter falls in an interval) instead of a finite set as in our current model.
For the former, the worst-case analysis might be relevant, and an even more challenge task arises when the leader is also trying to learn the prior distribution of follower types from the interaction.
For the latter, even the follower's problem of how to optimally misreport his type does not appear to have an obvious solution, unlike in our current model (where the follower simply needs to enumerate all types in the finite type set).
Another natural setting is that the follower does not need to keep imitating the fake type, switching back to playing their actual best response after the leader finishes learning. This is yet another generalization of the standard Bayesian Stackelberg game model, and the preliminary thought is that the same policy-based approach is also able to guarantee the leader as least her utility in the Bayesian Stackelberg equilibrium (similar to our Proposition~\ref{prp:icexists}).

Besides focusing on normal-form Stackelberg game models, it would also be interesting to position this work in a specific application area, such as security games where follower deception might be more common, yet more harmful, because of the adversarial nature of the game.

\bibliographystyle{named}
\bibliography{deception}

\end{document}